\documentclass[letterpaper]{article}
\usepackage{amsmath, amsthm, amssymb, verbatim,fullpage,color}
\usepackage{hyperref}
\usepackage{enumerate}

\usepackage[ruled,vlined,linesnumbered]{algorithm2e}
\DontPrintSemicolon

\usepackage[]{todonotes}

\newtheorem{theorem}{Theorem}[section]
\newtheorem{lemma}[theorem]{Lemma}
\newtheorem{proposition}[theorem]{Proposition}
\newtheorem{corollary}[theorem]{Corollary}
\newtheorem{claim}[theorem]{Claim}

\newtheorem{definition}[theorem]{Definition}

\newcommand{\junk}[1]{}

\DeclareMathOperator{\diam}{diam}
\newcommand{\lev}{L}
\newcommand{\levdiam}{\Delta}
\DeclareMathOperator{\degree}{deg}
\DeclareMathOperator{\emd}{cost}

\newcommand{\alg}{\mathcal{A}}

\renewcommand{\qed}{\nobreak \ifvmode \relax \else
      \ifdim\lastskip<1.5em \hskip-\lastskip
      \hskip1.5em plus0em minus0.5em \fi \nobreak
      \vrule height0.75em width0.5em depth0.25em\fi}

\newcommand{\Z}{\mathbb Z}
\newcommand{\R}{\mathbb R}
\newcommand{\N}{\mathbb N}

\newcommand{\eps}{\epsilon}

\def\E{\mathop{\mathbb{E}}\displaylimits}
\def\poly{\mathop{\rm{poly}}\nolimits}

\newcommand{\tour}{T^\circ}
\newcommand{\topt}{T^*}

\newcommand{\cost}[1]{\rho(#1)}

\newcommand{\boxapprox}{\gamma}

\newcommand{\mycell}{C}
\newcommand{\randpart}{\mathcal P}
\newcommand{\detpart}{P}
\newcommand{\myparagraph}[1]{\smallskip\noindent{\bf #1}}

\newcommand{\aanote}[1]{}
\newcommand{\konote}[1]{}
\newcommand{\snnote}[1]{}

\DeclareMathOperator*{\argmin}{arg\,min}

\title{Parallel Algorithms for Geometric Graph Problems}

\author{Alexandr Andoni
\\
Microsoft Research
\and
Aleksandar Nikolov
\\
Rutgers University
\and
Krzysztof Onak\thanks{Supported by the Simons Postdoctoral Fellowship. Research initiated while at CMU.}
\\
IBM TJ Watson
\and 
Grigory Yaroslavtsev
\\
ICERM, Brown University}

\begin{document}


\maketitle

\begin{abstract}
  We give algorithms for geometric graph problems in the modern
  parallel models such as MapReduce \cite{DG04-mapreduce,
    KSV10-mapreduce, GSZ11-sorting, BKS13-communication}. For example,
  for the Minimum Spanning Tree (MST) problem over a set of points in
  the two-dimensional space, our algorithm computes a
  $(1+\eps)$-approximate MST. Our algorithms work in a {\em constant}
  number of rounds of communication, while using total space and
  communication proportional to the size of the data (linear space and
  near linear time algorithms). In contrast, for general graphs,
  achieving the same result for MST (or even connectivity) remains a
  challenging open problem \cite{BKS13-communication}, despite drawing
  significant attention in recent years.

  We develop a general algorithmic framework that, besides MST, also
  applies to Earth-Mover Distance (EMD) and the transportation cost
  problem.  
  Our algorithmic framework has implications beyond the MapReduce
  model. For example it yields a new algorithm for computing EMD cost
  in the plane in near-linear time, $n^{1+o_\eps(1)}$. We note that
  while recently \cite{SA12-emd} have developed a near-linear time
  algorithm for $(1+\eps)$-approximating EMD, our algorithm is
  fundamentally different, and, for example, also solves the
  transportation (cost) problem, raised as an open question in
  \cite{SA12-emd}.
Furthermore, our algorithm immediately gives a $(1+\eps)$-approximation
algorithm with $n^{\delta}$ space in the streaming-with-sorting model
with $1/\delta^{O(1)}$ passes. As such, it is tempting to conjecture
that the parallel models may also constitute a concrete playground in
the quest for efficient algorithms for EMD (and other similar
problems) in the vanilla {\em streaming model}, a well-known open
problem \cite{streaming-open, Bertinoro11-open}.



\end{abstract}

\thispagestyle{empty}
\newpage
\setcounter{page}{1}

\section{Introduction}
\aanote{Afrati, Ullman, 2010, 2012}
\aanote{model like Mulct party communication, except for bounded fan-in}
Over the past decade a number of parallel systems have become widely
successful in practice.  Examples of such systems include MapReduce
\cite{DG04-mapreduce, DG08-mapreduce}, Hadoop \cite{white2012hadoop},
and Dryad \cite{isard2007dryad}.  Given these developments, it is natural to revisit
algorithmics for parallel systems and ask what new algorithmic or
complexity ideas Theoretical Computer Science can contribute to this
line of research (and engineering) efforts.

Two theoretical questions emerge: 1) What models capture well the capabilities of the
existing systems? 2) What new algorithmic ideas can we develop for
these models?  Addressing Question~1, researchers \cite{FMSSS10-mad,
  KSV10-mapreduce, GSZ11-sorting, BKS13-communication} have proposed a
model which balances simplicity and relevance to practice. We describe
this model later in Section \ref{sec:model}. As for Question~2, while there
already exist a few algorithms adapted or designed for this model (see
Section \ref{sec:comparison}), we feel that many more powerful
algorithmic ideas are still waiting to be developed.

The first natural question to ask is: do we really need new
algorithmics here? After all, we have had a lot of fundamental
research done on parallel algorithms in 1980s and 1990s, most notably in the
PRAM model, which one may hope to leverage to new models. Indeed,
the works of \cite{KSV10-mapreduce, GSZ11-sorting} have shown that one
can simulate PRAM algorithms in MapReduce with minimal slow-down. So
what is new?

The answer is that the parameters of the new models are such that we
can hope for {\em faster} algorithms than those possible in the PRAM
model.  The models allow for interleaving parallel and sequential
computation: in a single step, a machine can perform arbitrary
polynomial time computation on its local input; the time cost of the
algorithm is then measured in the number of \emph{rounds of
  communication} between machines. This makes it possible to achieve
{\em constant} parallel time for interesting problems, while in the
PRAM model functions that depend on the entire input generally require
logarithmic or larger parallel time.  For example, even computing the
XOR of $n$ variables requires near-logarithmic parallel-time on the
most powerful CRCW PRAMs \cite{BH89-CRCWlb}. In contrast, in the new
models, which are similar to a $n^\alpha$-fan-in circuit, one can
trivially solve XOR in $O(1/\alpha)$ parallel time. Indeed, the
MapReduce models rather fall under the blanket of the generic
Bulk Synchronous Parallel (BSP) model \cite{Val90-BSP}, though this
model has a number of parameters, and as such has not been thoroughly
explored. In particular, few solutions to even very fundamental
problems are known in the BSP models (see, e.g.,
\cite{Goodrich99-sort} for a sorting algorithm). The new models
instead focus on a specific range of parameters and tradeoffs, making
analysis more tractable.

The previous work on MapReduce models identifies a
captivating challenge problem: connectivity in a sparse graph. While
this problem has a classic logarithmic time PRAM algorithm
\cite{SV82-connect} 
we do not know whether we can solve it faster in the new models 
\cite{KSV10-mapreduce}. For this particular problem, though,
recent results show logarithmic lower bounds for restricted algorithms
\cite{BKS13-communication},  suggesting that the negative answer may
be more plausible.

\paragraph{Synopsis of contributions.} In this work, we focus on basic graph
problems in the {\em geometric} setting, and show we can achieve
$1+\epsilon$ approximation in a constant number of rounds. In fact, we
develop a common algorithmic framework applicable to graph questions
such as Minimum Spanning Tree and Earth-Mover Distance.  Thus, while
it may be hard to speed up standard graph algorithms (without
geometric context) in MapReduce-like models
\cite{BKS13-communication}, our results suggest that speedups can be
obtained if we manage to represent the graph in a geometric fashion
(e.g., in a similarity space).

Our framework turns out to be quite versatile, and, in fact has
implications beyond parallel computing. For example it yields a new
algorithm for computing EMD (cost) in the plane in near-linear time,
$n^{1+o_\eps(1)}$. We note that while recently \cite{SA12-emd} have
developed a near-linear time algorithm for $(1+\eps)$-approximating
EMD, our algorithm is different, and, for example, also solves the
transportation (cost) problem, raised as an open question in
\cite{SA12-emd}. In particular, our algorithm uses little of the combinatorial
structure of EMD, and essentially relies only on an off-the-shelf LP
solver. In contrast, \cite{SA12-emd} intrinsically exploit the
combinatorial structure, together with carefully designed data
structures to obtain a $O_\eps(n\log^{O(1)}n)$ time algorithm. Their approach, however,
seems hard to parallelize.

Another consequence is also a $1+\eps$ approximation algorithm in the
{\em streaming-with-sorting} model, with $n^{\delta}$ space and
$1/\delta^{O(1)}$ passes. Hence, our EMD result suggests the new
parallel models as a concrete playground in the quest for an efficient
streaming algorithm for EMD, a well-known open question
\cite{streaming-open, Bertinoro11-open}.

\subsection{The Model}
\label{sec:model}

We adopt the most restrictive MapReduce-like model among
\cite{KSV10-mapreduce, GSZ11-sorting, BKS13-communication} (and BSP
\cite{Val90-BSP} for a specific setting of parameters). Following
\cite{BKS13-communication}, we call the model {\em Massively Parallel
  Communication} or MPC (although we explicitly consider the local
sequential running times as well).

Suppose we have $m$ machines (processors) each with space $s$,
where $n$ is the size of the input and $m\cdot s=O(n)$. Thus, the
total space in the system is only a constant factor more than the input
size, allowing for minimal replication.

The computation proceeds in rounds. In each round, a machine performs
local computation on its data (of size $s$), and then sends messages
to other machines for the next round. Crucially, the total amount of
communication sent or received by a machine is bounded by $s$, its
space. For example, a machine can send one message of size $s$, or $s$
messages of size $1$. It cannot, however, broadcast a size $s$ message to
every machine. In the next round, each machine treats the received
messages as the input for the round.

The main complexity measure is the number of rounds $R$ required to
solve a problem, which we consider to be the ``parallel time'' of the
algorithm. Some related models, such as BSP, also consider the
sequential running time of a machine in a round. We will de-emphasize this
aspect, as we consider the information-theoretic question of
understanding the round complexity to be the first-order business
here. In particular, the restriction on space alone (i.e., with {\em
  unbounded} computation per machine) already appears to make certain
problems take super-constant number of rounds, including the
connectivity in sparse graphs. Nevertheless it is natural to minimize
the local running time, and indeed our (local) algorithms run in time
polynomial in $s$, leading to $O(ns^{O(1)})$ overall work.

What are good values of $s$ and $R$? As in \cite{KSV10-mapreduce,
  GSZ11-sorting}, we assume that space $s$ is polynomial in $n$, i.e.,
$s=n^\alpha$ for some $\alpha>0$. We consider this a justified choice
since even under the natural assumption that $s\ge m$ (i.e., each
machine has an index of all other machines), we immediately obtain
that $s\ge \sqrt{n}$.\footnote{Furthermore, it is hard to imagine a
  data set where $\sqrt[3]{n}$ is larger than the memory of a commodity
  machine.}

Our goal is to obtain $R=\poly(\log_s n)=O(1)$ rounds. Note that we do not
hope to do better than $O(\log_s n)$ rounds as this is required even
for computing the XOR of $n$ bits.

Finally, note that the total communication is, {\em a fortiori}, $O(n)$
per round and $O(nR)=O(n)$ overall.

\myparagraph{Streaming models.} The above MPC model essentially resides in 
between two streaming models.

First, it is at least as strong as the ``linear streaming'' model,
where one stores a (small) linear sketch of the input:
if one has a linear sketch algorithm using space $s$ and $R$ passes,
then we also have a parallel algorithm with local space $s^2$ (and
$m=O(n/s^2)$ machines) and $O(R \log_s m)$ rounds.

Second, the above model can be simulated in the model of streaming
{\em with a sorting primitive} \cite{ADRR04-streamSort}. The latter
model is similar to the standard multi-pass streaming model, but allows
for both annotating the stream with keys as we go through it and
sorting the entire stream according to these keys.  In particular,
sorting is considered in this model to be just another pass.  Then if
we have a parallel algorithm with $s$ space and $R$ rounds, we also
obtain a streaming-with-sorting algorithm with $O(s)$ space and $O(R)$
passes.



\subsection{Our Results}

In this work, we focus on graph problems for geometric graphs. We
assume to have $n$ points immersed in a low-dimensional space, such as
$\R^2$ or a bounded doubling dimensional metric. Then we consider
the complete graph on these points, where the weight of each edge is
the distance between its endpoints.\footnote{Since our algorithms work
  similarly for norms such as $\ell_1,\ell_2,\ell_\infty$,
  we are not specific about the norm.}

We give parallel algorithms for the following problems:
\begin{itemize}
\item
Minimum Spanning Tree (MST): compute the minimum spanning tree on the
nodes. Note that MST is related to the hierarchical agglomerative
clustering with single linkage, a classic (and practical) clustering
algorithm \cite{zahn1971graph, kleinberg2006algorithm}.

We show how to compute a $1+\eps$ approximate MST over $\R^d$ in
$O(\log_s n)$ rounds, as long as $(1/\eps)^{O(d)}<s$. Note that the
number of rounds does not depend on $\eps$ or $d$.  We extend the
result to the case of a general point set with doubling dimension
$d$. All our algorithms run in time $s^{O(1)}(1/\eps)^{O(d)}$ per
machine per round. 

We note that our algorithm outputs the actual tree (not just the cost)
of size $\approx n$, meaning in particular that the output is also
stored in a distributed manner. The algorithms appear in Section
\ref{sect:mst} and \ref{sec:doubling}.

\item
Earth-Mover Distance\footnote{Also known as min-cost bichromatic
  matching, transportation distance, Wasserstein distance, and
  Kantorovich distance, among others} (EMD): given an equipartition
of the points into red and blue points, compute the min-cost red-blue
matching. A generalization is the \emph{transportation distance},
in which red and blue points have positive weights of the same total
sum, and the goal is to find a min-cost matching between red and
blue masses.  EMD and its variants are a common
metric in image vision \cite{RTG, GD-kernel}.

We show how to approximate the EMD and transportation cost up to a
factor of $1+\eps$ over $\R^2$ in $(\log_s n)^{O(1)}$ rounds, as long
as $(\log n)^{(\eps^{-1}\log_s n)^{O(1)}}<s$. The running time per machine
per round is polynomial in $s$. Note that, setting $s=2^{\log^{1-c}
  n}$ for small enough $c>0$, we obtain a (standard) algorithm with
overall running time of $n^{1+o(1)}$ for any fixed $\eps>0$. Our algorithm
can also be seen as an algorithm in the streaming-with-sorting model,
achieving $n^{\delta}$ space and $1/\delta^{O(1)}$ rounds by setting
$s=n^\delta$. Our algorithm does not output the actual matching (as
\cite{SA12-emd} do).  The algorithm appears in Section \ref{sec:emd}.
\end{itemize}

All our algorithms fit into a general framework, termed
Solve-And-Sketch, that we propose for such problems. The framework is
naturally ``parallelizable'', and we believe is resilient to minor
changes in the parallel model definition. We describe the general
framework in Section \ref{sec:preliminaries}, and place our algorithms
within this framework. The actual implementation of the framework in
the MPC model is described in Section \ref{sec:implementation}.

It is natural to ask whether our algorithms are optimal. Unfortunately,
we do not know whether both approximation and small
dimension are required for efficient algorithms. However, we show that
if we could solve exact MST (cost) in $l_\infty^{O(\log n)}$, we could
also solve sparse connectivity (in general graphs), for which we have
indications of being impossible~\cite{BKS13-communication}. We also
prove a query-complexity lower bound for MST in spaces with bounded
doubling dimension in the black-box distance oracle model. In this
setting, both approximation and dimension restriction seem necessary.
These results appear in Section~\ref{sec:lb}.

\subsection{Motivation and Comparison to Previous Work}
\label{sec:comparison}

\myparagraph{The model perspective.} 
\cite{KSV10-mapreduce} have initiated the study of {\em dense} graph
problems in the MapReduce model they define, showing constant-round
algorithms for connected-components, MST, and other problems. In the
dense setting, the parameters are such that $m \gg s\gg n$,
where $n$ is the number of vertices and $m$ is the number of edges. In this case, the solution (the size of which is $O(n)$) fits on a single
machine.

In this regime, the main technique is filtering (see also
\cite{LMSV11-filtering}), where one iteratively sparsifies the input
until the entire problem fits on one machine, at which moment the
problem is solved by a regular sequential algorithm. For example,
for connected-components, one can just throw out edges locally,
preserving the global connectivity, until the graph has size at most
$s$.

Somewhat departing from this is the work of \cite{EIM11-clustering},
who give algorithms for $k$-median and $k$-center, using
$s=O(k^2n^\delta)$. Instead of filtering, they employ (careful)
sampling to reduce the size of the input until it fits in one machine
and can be solved sequentially. Note that, while the entire
``solution'' is of size $n\gg s$, the solution is actually represented
by $k\ll s$ centers. \cite{KMVV13-greedy} further generalize both the
filtering and sampling approaches for certain greedy problems, again
where the solution , of size $k\ll s$, is computed on a single
machine at the end. 

Also highly relevant are the now-classic results on {\em coresets}
\cite{AHV-survey, feldman2011unified}, which are generic
representation of (subset of) input with the additional property of
being mergeable. These are often implementable in the MapReduce model
(in fact, \cite{EIM11-clustering} can be seen as such an
implementation). However, coresets have been mostly used for {\em
  geometric problems} (not graph problems), which often have a small
solution representation.

We contrast the ``dense'' regime with the ``sparse'' regime where
$s$ is much smaller than the size of the solution.
Most notably, for the problem of computing the
connected components in a sparse graph, we have no better
algorithm than those following from the standard PRAM literature, despite
a lot of attention from researchers.
In fact, \cite{BKS13-communication} suggest it
may be hard to obtain a constant parallel-time for this problem.

Our algorithms rather fall in the ``sparse'' regime, as the solution
(representation) is larger than the local space $s$. As such, it appears hard
to apply filtering/sampling technique that drops part of the input
from consideration. Indeed, our approach can be rather seen as a
generalization of the notion of coreset.

We also mention there are other related works in MapReduce-like models,
e.g., \cite{CKT10-maxcover, BPT11-setcover, BKV12-densest}, but which
achieve at least a logarithmic parallel time.

\myparagraph{The problems perspective.}
While we are not aware of a previous study of geometric graph problems
in the MapReduce models, these problems have been studied extensively
in other standard models, including 1) near-linear time algorithms, and 2)
streaming algorithms.

Linear time (approximate) algorithms for MST are now classic results
\cite{Vaidya88-mst, CK93-mst}.
For EMD, it is only very recently that researchers found a near-linear
time approximation algorithm \cite{SA12-emd} (following a line of work
on the problem \cite{Vai89-EMD, AES00-EMD, VA99-EMD, AV04, I07,
  SA12-logEps}). Our framework naturally leads to near-linear time
algorithms.

In the streaming model, a generic approach to approximate a large
class of geometric graph problems has been introduced in
\cite{I04}. The work of \cite{I04} has generally obtained logarithmic
approximation for many problems and subsequently there has been a lot
of research on improving these algorithms. Most relevantly,
\cite{FIS-mst} have shown how to $1+\eps$ approximate MST {\em
  cost}. We note that their algorithm outputs the cost only and does
not lead to an algorithm for computing the actual tree as we
accomplish here.

Getting $1+\eps$ approximate streaming algorithm for EMD is a
well-known open question \cite{streaming-open, Bertinoro11-open}. The
best known streaming algorithm obtains a $O(1/\delta)$ approximation
in $n^{\delta}$ space for any $\delta>0$ \cite{ADIW-EMD}.

Our algorithmic framework immediately leads to an algorithm for
computing $1+\epsilon$ approximation in $n^{\delta}$ space and
$1/\delta^{O(1)}$ passes in the streaming-with-sorting model.  In general,
our EMD result implies one of the following: either 1) it illustrates
the new parallel models as a concrete mid-point in the quest for an
efficient streaming algorithm for EMD, or 2) it separates the new
parallel models from the linear streaming model, showing them as
practical models for sublinear space computation which are strictly
more powerful than streaming. We do not know which of these cases is
true, but either would be an interesting development in the area of
sublinear algorithms.



\subsection{Techniques}
\label{sec:techniques}

We now describe the main technical ideas behind our algorithms.
Our MST algorithm is simple, but requires some careful analysis, while the EMD algorithm is technically
the most involved.

\myparagraph{MST.}
To illustrate the main ideas involved in the algorithm it suffices to consider the problem over the 2D grid $[0,n]^2$.  The framework
consists of three conceptual parts: partition, local solution, and
sketch. The partition will be a standard quadtree decomposition,
where we impose a hierarchical grid, randomly shifted in the space. In
particular, each cell of the grid is recursively partitioned into
$\sqrt{s}\times \sqrt{s}$ cells, until cell size is $\sqrt{s}\times
\sqrt{s}$. The partition is naturally represented by a tree of arity
$s$.

The other two parts are the crux of the algorithm. Consider first the
following recursive na\"ive algorithm. Going bottom-up from the leaves
at every cell in the quadtree we compute the minimum spanning tree
among the input points (local solution), and then send a sketch of this tree to the
upper-level cell.  The problem is solved recursively in the
upper-level cell by connecting partial trees obtained from the lower
level.

However, such an algorithm does not yield
a~$(1+\epsilon)$-approximation. While constructing minimum spanning
tree in a cell, the limited local view may force us to make an
irrevocably bad decision. In particular, we may connect the nodes in
the current cell, which in the optimum solution are connected by a
path of nodes outside the cell.  Consider an example on
Figure~\ref{fig:mst-bad-example}. If the four points in the MST
instance form the corners of a 2x1 rectangle then with constant
probability (over the choice of the random partition), the edges of
length 1 will be cut by a higher level of the partition than the edges
of length 2.  If an algorithm commits to constructing a minimum
spanning tree for subproblems, then the two solid edges will be
selected before the dashed edges are considered. At the next level the
algorithm will select one of the dashed edges, and the total cost of
the tree will be 5. However, substituting the green edge for the red
edge results in a tree of cost 4.

\begin{figure}[ht]
  \hspace{0.7mm}
\begin{minipage}[b]{0.4\linewidth}
  \includegraphics[width=\textwidth]{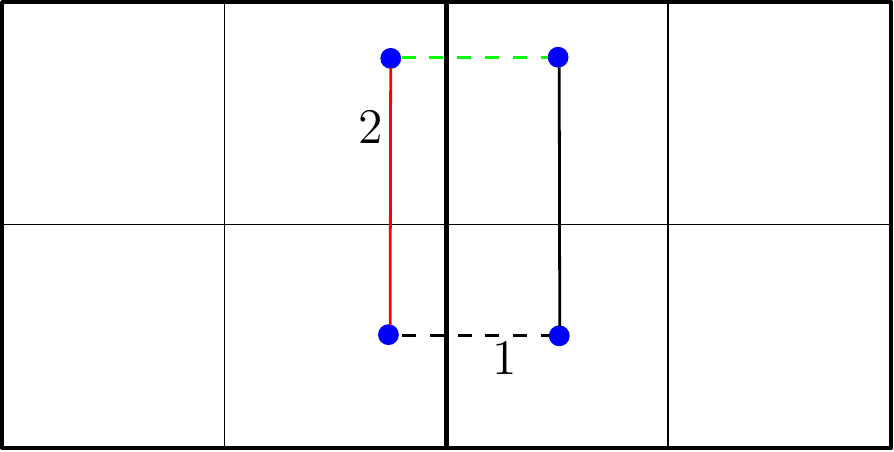}
  \caption{A bad example for the na\"ive MST algorithm.}
  \label{fig:mst-bad-example}
\end{minipage}
\hspace{0.15\linewidth}
\begin{minipage}[b]{0.4\linewidth}
  \includegraphics[width=\textwidth]{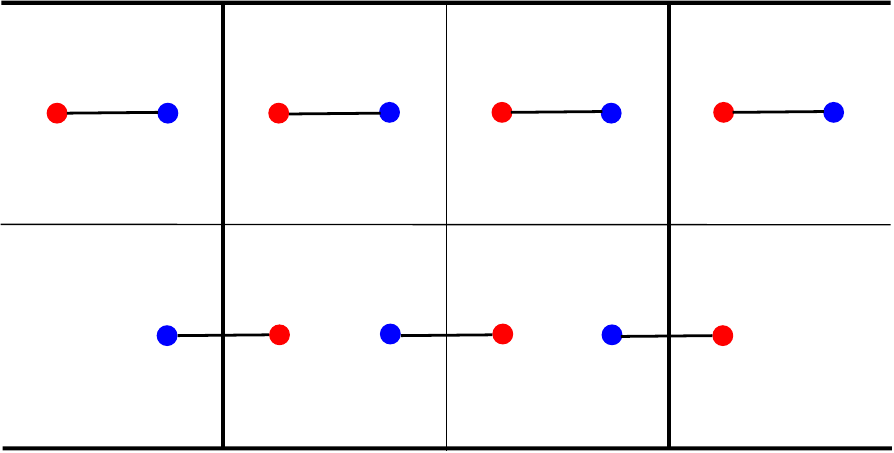}
  \caption{Local view is insufficient for EMD.}
\label{fig:emd-bad-ex}
\end{minipage}
\end{figure}

The challenge is to produce a local solution, without committing to
decisions that may hurt in the future.  To accomplish this, our local
solution at a cell is to find the minimum spanning forest among the
input points, using only {\em short edges}, of length at most $\eps$
times the diameter $\Delta$ of the cell. Note that it is possible that
the local set of points remains disconnected.

Our sketch for each cell consists of an $\eps^2\Delta$-net\footnote{An
  $r$-net of a point set is the maximal subset with pairwise distances at
  least $r$.} of points in the cell together with the information
about connectivity between them in the current partial solution.  Note
that the size of the sketch is bounded by $O(\eps^{-4})$.  This sketch
is sent to the parent cell.  The local solution at the parent node now
will consist of constructing a minimum spanning forest for the connected
components obtained from its children.

In the analysis we argue that our algorithm is equivalent to Kruskal's
algorithm, but run on a graph with modified edge weights. We carefully
construct the edge weights recursively to follow the execution of the
algorithm. To prove the approximation guarantee, we make sure that the
modified edge weights do not differ much from the original weights
distances, in expectation, over the initial random
shift.

We also generalize our algorithm to the case of a point set with
bounded doubling dimension. Here, the new challenge is to construct
a~good hierarchical partition.

\myparagraph{EMD.}
Our EMD algorithm adopts the general principle from MST, though the
``solution'' and ``sketch'' concepts become more intricate. Consider
the case of EMD over $[n]^2$.  As in MST, we partition the space using
a hierarchical grid, represented by a tree with arity $s$.

In contrast to the MST algorithm, there are no local ``safe''
decisions one can make whatsoever. Consider
Figure~\ref{fig:emd-bad-ex}. The two rows of points are identical
according to the local view. However, in one case we should match all
points internally, and in the other, we should leave the end points to
be matched outside. As far as the local view is concerned, either
partial solution may be the right one. If we locally commit to the
wrong one, we are not able to achieve a $1+\eps$ approximation no
matter what we do next. This lack of ``partial local solution'' is
perhaps the reason why EMD appears to be a much harder problem than
even non-bipartite Euclidean matching, for which efficient algorithms
have been known for a while now \cite{A98-TSP}.

It seems that we need to be able to sketch the {\em entire
  set of local solutions}. In particular, in the above case, we want
to represent the fact that both (partial) matchings are valid
(partial) solutions---as a function of what happens outside the local
view---and include both of them in the sketch. Note that the two
solutions from above have different ``interfaces'' with the outside
world: namely the second one leaves two points to be matched to the
outside. 

And this is what we accomplish: we sketch the set of {\em all}
possible local solutions. While reminiscent of the dynamic programming
philosophy, our case is burdened by the requirement to represent
this in sublinear (local) space. This seems like a daunting task:
there are just too many potentially relevant local solutions. For
example, consider a cell where we have $n/3$ red points close to the
left border and $n/3$ blue points close to the right border. If there
are some $k$ blue points to the left of the cell, and $k$ red ones to
the right, then, inside the cell we should match exactly $n -k $ pairs
of points. Hence, from a local viewpoint of the cell (which does not
know $k$), there are $n/3$ potentially relevant local
solutions. Sketching each one of them already takes space $\Omega(n)\gg s$.

Our algorithm manages to sketch this set of relevant local
solutions approximately. Suppose we define a function $F$ of $d$
coordinates, one for ``each position'' in the local cell. In
particular, $F$ takes as argument a vector $x\in \Z^d$ that specifies,
for each $i\in[d]$, how many points are left unmatched at position
$i$, with the convention that positive $x_i$ signifies red points and
negative $x_i$ signifies blue points. Then we can define $F(x)$ to be
the cost of the optimal matching of the rest of the points (with
points specified by $x$ excluded from the local matching).

It would great to sketch the function $F:\Z^d\to \R_+$. Suppose we 
even reduced $d$ to be effectively $1/\eps^2$ (it suffices to consider
only positions at the $\eps$-net of the cell, similarly to what happens
in MST).

Nevertheless, even if $d=1/\eps^2$, we do not know how to sketch this
function $F$. For example, even for a function $F$ of two parameters
which is guaranteed to be monotone, convex, and Lipschitz, a concise
sketch is generally not possible (in our case, $F$ is in fact not even
monotone).  What we show instead is that we can sketch the function
$F'(x)=F(x)+\|x\|_1\cdot A$, for some convenient factor $A$. It turns
out that the additional term of $\|x\|_1\cdot A$ is tolerable as it
will capture part of the cost of matching {\em at the next level
  up}. Then a sketch of $F'$ can just consist of $F'(x)$ values at
$(\eps^{-1}\log n)^{O(d)}$ well-chosen $x$'s.

These ideas eventually lead to an {\em information-theoretic}
algorithm for EMD, namely with the promised guarantees on space and
rounds. It remains to make the running time of the local step polynomial in
$s$. While computing $F'(x)$ at a leaf is straight-forward (it's
essentially a matching), it is less clear for an internal node, where
we have to compute $F'(x)$ from the approximate sketches of $F'_j$ of
the node's children $j\in J$.

To compute $F'$ in polynomial time, we find the largest convex
function which agrees with the sketch of each $F'_j$; this gives a set
of functions that are ``piece-wise linear'' and can be easily absorbed
into a larger LP to compute $F'$ at the internal node. We can do so
because $F'_j$ is a convex function, so the largest convex function
that agrees with its sketch is sandwiched between the sketch itself
and the actual $F'_j$. In the end, the local running time is polynomial in
$s$, because the resulting LPs can be solved with an arbitrary
polynomial time LP algorithm.

We note that the total running time (work) is $n^{1+o(1)}$ (when setting 
$s=\log^{1-O(1)}n$),
but we hope this can be brought down to $\tilde O(n)$ by exploiting
more of the combinatorial structure of the problem and sketching $F$
or $F'$ in space polynomial in $d$ (instead of exponential in $d$ as
we do here).

\subsection{Some Challenges For Future Research}

We note that many research question remain open and may be relevant to
both MapReduce/MPC-like models, as well as more generally to the area of
sublinear algorithms. We list a few:
\begin{itemize}
\item
Can we sketch the EMD partial solution(s) using $(\eps^{-1}\log
n)^{O(1)}$ space? Can we compute the actual matching as well? This may
lead to a $n(\eps^{-1}\log n)^{O(1)}$ overall time (work) for EMD and
transportation problems (assuming the local running time is polynomial).
\item
More ambitiously, can these ideas lead to an efficient streaming
algorithm for EMD, solving the open question of
\cite{streaming-open, Bertinoro11-open}? When do
algorithms in our framework lead to streaming algorithms?
\item
Lower bounds: Is $1+\eps$ approximation or constant dimension required
for geometric graph problem such as MST or EMD? What techniques need to be developed to prove such lower bounds?
\item
Can we solve data-structure like problems? In particular, some new
systems allow for incremental updates to the input, with
the expectation that the new computation be minimal (see, e.g.,
\cite{murray2013naiad})?
\end{itemize}




\section{Preliminaries: Solve-And-Sketch Framework}
\label{sec:preliminaries}

We now introduce the framework for our algorithms, termed
Solve-And-Sketch. Its main purpose is to identify and decouple the
crux of the algorithm for the specific problem from the implementation
of the algorithm in the parallel model such as MPC.

The framework requires a ``nice'' hierarchical partition of the
space. We view the hierarchical partition as a tree, where the arity
is upper bounded by $\sqrt{s}$, and the depth is $O(\log_s n)$. The
actual computation is broken down into small ``local computational
chunks'', arranged according to the hierarchical partitioning. The
computation proceeds bottom-up, where at each node, the input (from
below) is processed and the results are compressed into a small {\em
  sketch} that is sent up to the parent. Each level of the tree will
be processed in parallel, with each node assigned to a machine.

In particular, the local computation at a node, termed ``unit step'',
consists of two steps:

\begin{enumerate}
\item[{\bf Solve:}] Given the local inputs, we compute the set of
  partial or potential solutions. For leaves, the local information
  consists of the points in that part, and for internal nodes, it is
  the information obtained from the children.
\item[{\bf Sketch:}] Sketch the partial solution(s), using total space
  at most $p_u\le\sqrt{s}$, and send this up the tree to the parent as a
  representation of solution(s) in this part.
\end{enumerate}


The main challenges will be how to: 1) compute the partition, 2) define the right
concept of ``local solutions'' in a part, and 3) sketch this concept
as a sufficient representation of all potential solutions in this
part. Often the na\"ive choice of the a local solution cannot be used,
because it either ignores global information in a way that can damage
the optimality of the algorithm, or it cannot be represented in
sublinear space.

We now define more formally the notions of partition and of a unit
step.

\myparagraph{Hierarchical Partition.}  We use a hierarchical partition
for inputs in $(\R^d, \ell_2)$ that is an analogue of a
randomly-shifted quad-tree but a {\em high} branching factor (rather
than the usual $2^d$). We denote the branching factor by $c$.  We
describe this partitioning scheme next (see
Section~\ref{sec:eucl-part} for additional details). The
partitions we use to compute MST in a low doubling dimension metric space
are more involved; see Section~\ref{sec:doubling}.

We assume that the points have integer coordinates in the range $[0,
  \Delta]$, where $\Delta=n^{O(1)}$. We show how to remove this
bounded aspect ratio assumption in Section \ref{sec:aspectRatio}.  Let
$v \in \mathbb R^d$ be a vector chosen uniformly at random from
$(-\Delta, 0]^d$. We construct a hierarchical partition, denoted
  $\detpart = (\detpart_0, \ldots, \detpart_\lev)$. The top level
  $\detpart_\lev$ has a single part containing the whole input, and is
  identified with the cube $\{x: \forall i\ v_i \leq x_i \leq v_i +
  2\Delta\}$. Then we construct $\detpart_{\ell-1}$ from
  $\detpart_\ell$ by subdividing each cube associated with a part in
  $\detpart_{\ell}$ into the $c$ equal sized cubes (via a grid with
  side-length $c^{1/d}$), thus creating a part associated with each
  smaller cube. In the final level $\detpart_0$, each part is a
  singleton, i.e.~all associated cubes contain at most a single point
  from the input. Since we assumed all points have integer coordinates
  in $[0, \Delta]$, it is enough to take $\lev = d\log_c
  \Delta=O(d\log_c n)$.  For each level partition $\detpart_\ell$, we
  call its parts as ``cells''. For a cell $C \in \detpart_\ell$, we
  can consider the subdivisions of the cell $C$ into next-level cells,
  i.e., all $C' \in \detpart_{\ell-1}$ such that $C' \subseteq C$,
  which we call the ``child cells'' of $C$. For our implementation, we
  also need to label each child cell of $C$ with an integer in
  $[c]$. We can do this in a number of ways, for example, by
  lexicographic order on coordinates of the center of each cube
  associated with a child cell.

\myparagraph{Unit Step.}  The other important component of the sketch
and solve framework is the unit step, which is an algorithm $\alg_u$
that is applied to each cell $C \in \detpart_\ell$ for $\ell = 1,
\ldots, \lev$. At level $1$, $\alg_u$ takes as input the points in
$C$, and at level $\ell > 1$, $\alg_u$ takes as input the union of
outputs of the unit steps applied to the children of $C$. The output
of $\alg_u$ on the top-most cell $P_L$ is the output of the problem
(perhaps after some post-processing). We define functions $p_u, t_u,s_u$
as follows: on input of size $n_u$, $\alg_u$ produces an output of
size at most $p_u(n_u)$, runs in time at most $t_u(n_u)$, and uses a
total space $s_u(n_u)$. We require that, on empty input, $\alg_u$
produces empty output. We call the algorithm that applies $\alg_u$ to
each cell of the partition in the above fashion the
\emph{Solve-And-Sketch} algorithm.

We prove that once we have a unit step algorithm for a problem, we
also obtain a complete parallel algorithm for the said problem. Hence
designing the unit step for a problem is the crux for obtaining a
parallel algorithm and is decoupled from the actual implementation
specifics in the considered parallel model.

\begin{theorem}[Solve-And-Sketch]
\label{thm:unitstep}
Fix space parameter $s=(\log n)^{\Omega(d)}$ of the MPC model. Suppose
there is a unit step algorithm using local time $t_u(n_u)$, space
$s_u(n_u)$, and output size $p_u(n_u)$ on input of size $n_u$. Assume
the functions $t_u,s_u,p_u$ are non-decreasing, and also satisfy:
$s_u(p_u(s))\le s^{1/3}$ and $p_u(s)\le s^{1/3}$. Then we can set
$c=s^{\Theta(1)}$ and $\lev=O(\log_s n)$ in the partitioning from
above, and we can implement the resulting Solve-And-Sketch algorithm
in the MPC model in $(\log_s n)^{O(1)}$ rounds. Local runtime is $s\cdot
t_u(s)\cdot (\log n)^{O(1)}$ (per machine per round).
\end{theorem}

The proof of the theorem is sensitive to the actual parallel
model hence we defer it, along with other details of implementation
in the MPC model, to Section \ref{sec:implementation}. In the sections
that follow, we describe how to implement the unit step algorithm for
the two considered problems, which, by Theorem \ref{thm:unitstep}, will
imply efficient MPC algorithm.

\section{Minimum Spanning Tree}\label{sect:mst}

In this section we prove the existence of an efficient MPC algorithm
that computes a spanning tree of a given point set in Euclidean space
of approximately minimal cost. 

\begin{theorem}
\label{thm:mstMain}
Let $\eps > 0$, and $s\geq (\eps^{-1}\log_sn)^{O(1)}$. Then there
exists an MPC algorithm that, on input a set $S$ in $\R^d$ runs in
$(\log_s n)^{O(1)}$ rounds and outputs a spanning tree of cost (under
the Euclidean distance metric $\ell_2^d$ for $d = O(1)$) at most $1 +
\eps$ factor larger than the optimal. Moreover, the running time per
machine is near linear in the input size $n_u$, namely
$O(n_u\eps^{-d}\log^{O(1)}n_u)$.
\end{theorem}

We prove the theorem above by exhibiting a unit step algorithm within
the Solve-And-Sketch framework from Section
\ref{sec:preliminaries}. In fact our unit step algorithm will work
with partitions more general than the quadtree-based partition
described in Section~\ref{sec:preliminaries}. This will allows us to
apply the unit step to point sets in low doubling dimension as well,
once we have constructed an appropriate hierarchical
partition. Details for the doubling dimension case appear in
Section~\ref{sec:doubling}.

\subsection{Hierarchical Partitions}

Let us first define some metric geometry preliminaries and then define
the general notion of partitions that our unit step algorithm uses. 

We denote a metric space on a ground set $S$ with a distance function $\rho(\cdot,\cdot)$ as $M(S, \rho)$. For $S' \subseteq S$ we denote the \textit{diameter} of $S'$ as $\levdiam(S') = \sup_{x,y \in S'} \rho(x,y)$.
A \textit{ball} in $M$ centered in $x \in S$ with radius $r$ is denoted as $B_M(x,r) = \{y \in S | \rho(x,y) \le r\}$. If the metric is clear from the context, we omit the subscript and write simply $B(x,r)$.

\begin{definition}[Coverings, packings, and nets]
Let $M=(S,\rho)$ be a metric space and let $\delta$ and $\delta'$ be positive reals.
A set $S' \subseteq S$ is a:
\begin{enumerate}
\item \emph{$\delta$-cover} if for any point $x \in S$, there is a point $y \in S'$ such that $\rho(x,y) \le \delta$,
\item \emph{$\delta'$-packing} if for any two points $x,y \in S'$, it holds that $\rho(x,y) \ge \delta'$,
\item \emph{$(\delta,\delta')$-net} if it is both a $\delta$-covering and a $\delta'$-packing,
\item \emph{$\delta$-net} if it is a $(\delta,\delta)$-net.
\end{enumerate}
\end{definition}

\begin{definition}[Doubling dimension]
The \emph{doubling dimension} of a metric space is the smallest $d$ such that for all $r \ge 0$, any ball of radius $r$ can be covered with at most $2^d$ balls of radius $r/2$.
\end{definition}

\begin{lemma}[Dimension of restricted space]\label{lem:ddim_restricted} 
Let $M_1=(S_1,\rho)$ be a metric space of doubling dimension $d$. Let $M_2 = (S_2,\rho)$ be a metric space such that $S_2 \subseteq S_1$. The doubling dimension of $M_2$ is at most $2d$.
\end{lemma}

\begin{proof}
Consider an arbitrary ball $B_{M_2}(p,r)$ in $M_2$. Clearly, $B_{M_2}(p,r) \subseteq B_{M_1}(p,r)$. By definition, $B_{M_1}(p,r)$ can be covered with at most $(2^d)^2$ balls $B_1$, \ldots, $B_k$ of radius $r/4$ in $M_1$. Now, for each ball $B_i$ such that $B_i \cap S_2 \ne \emptyset$, pick an arbitrary point $p_i \in B_i \cap S_2$.

We claim that the collection of balls $B_{M_2}(p_i,r/2)$ for all such $i$ covers $B_{M_2}(p,r)$. Consider a point $p'$ in $B_{M_2}(p,r)$. It belongs to a ball $B_i$. By the triangle inequality it holds $\rho(p',p_i) \le r/2$, and therefore, $p'$ belongs to
$B_{M_2}(p_i,r/2)$.

Summarizing, every ball $B_{M_2}(p,r)$ can be covered with at most $2^{2d}$ balls of radius $r/2$ in $M_2$, which implies that the doubling dimension of $M_2$ is at most $2d$.
\end{proof}

We use the following terminology to abstract out the common ideas behind our algorithms for Euclidean and bounded doubling dimension spaces.
The common component of both algorithms is randomized hierarchical partition of the input space $M(S,\rho)$ (see, e.g.~\cite{Talwar04}).
 A \textit{deterministic hierarchical partition} $\detpart$ with $\lev$ \textit{levels} is defined as a sequence $\detpart = (\detpart_0, \dots, \detpart_{\lev})$,
where $\detpart_L = \{S\}$ and each level $\detpart_{\ell}$ is a subdivision of $\detpart_{\ell + 1}$.
For a partition $\detpart_i$ we call its parts \textit{cells}.
The \textit{diameter} at level $i$ is  $\levdiam(\detpart_i) = \max_{C \in \detpart_i} \levdiam(C)$.
The degree of a cell $C \in \detpart_\ell$ is $\degree(C) = |\{C' \in \detpart_{\ell-1}: C' \subseteq C\}|$. 
The \textit{degree} of a hierarchical partition is the maximum degree of any of its cells.
We denote the unique cell at level $\ell$ containing a point $x$ as
$\mycell_\ell(x)$, i.e. $\mycell_\ell(x)$ is defined by $x \in
\mycell_\ell(x)$ and $\mycell_\ell(x) \in \detpart_\ell$. 
For $\ell' \le \ell$ and $\mycell \in P_{\ell'}$ we define $\mycell_\ell(\mycell)$ analogously as the unique cell in the level $\ell$ containing $C$.
We will also work with \textit{randomized hierarchical partitions} which we treat as distributions over deterministic hierarchical partitions.
We will denote such distributions as $\randpart$ to distinguish them from deterministic partitions. 

\begin{definition}[Distance-preserving hierarchical partition]\label{def:partition}
For $a, b \in (0,1)$, $\gamma > 1$ a randomized hierarchical partition $\randpart$ of a metric space $M(S,\rho)$ with $L$ levels is $(a,b)$-\textit{distance-preserving} with approximation $\boxapprox$ if every deterministic partition $P = (P_0, \dots, P_L)$ in its support satisfies the following properties for $\levdiam_\ell = \boxapprox a^{L - \ell} \levdiam(S)$:
\begin{enumerate}
\item (Bounded diameter) For all $\ell \in \{0, \dots, L\}$: 
$$\levdiam(\detpart_{\ell}) \le \levdiam_\ell.$$
\item (Probability of cutting an edge) For every $x,y \in S$:
$$\Pr[\mycell_\ell(x) \neq \mycell_\ell(y)] \le b\frac{\rho(x,y)}{\levdiam_{\ell}}.$$
\end{enumerate}
\end{definition}

To simplify the presentation we will refer to an $(a,b)$-distance preserving hierarchical partition as just $(a,b)$-partition.
The parameter $\boxapprox$ plays a less important role in our proofs so we omit it to simplify presentation.
Moreover, if an $(a,b)$-partition has degree $c$ we will call it just an $(a,b,c)$-partition. \junk{We also assume that for every cell $C$ in an $(a,b,c)$-partition its children are numbered by consecutive integers, starting from $1$. For a child $C'$ of $C$ we denote its index in such numbering as $h(C')$. }
An example of a randomized $(a,b,c)$-partition is a randomly shifted and rotated quadtree in the Euclidean space $(\mathbb R^d, \ell_2)$, which is a $(1/2, O(d) , 2^d)$-partition
The  Euclidean hierarchical partition described in
Section~\ref{sec:preliminaries} is an $(c^{-1/d}, O(d),
c)$-partition: see Section~\ref{sec:eucl-part} for a detailed
discussion.

\newcommand{\ccomp}{Q}
\newcommand{\supp}{supp}
\newcommand{\weight}{w_P}
\newcommand{\weighti}[1]{\weight^{#1}}

\subsection{The Unit Step Algorithm}

Our Solve-and-Sketch (SAS) algorithm for computing an approximate
minimum spanning tree (MST) works with a partition $\detpart =
(\detpart_0, \ldots, \detpart_\lev)$ of the input $M(S,\rho)$, sampled
from a randomized $(a, b, c)$-partition $\randpart$. Recall that a SAS
algorithm proceeds through $\lev$ levels, and in level $\ell$ a
\emph{unit step} algorithm is executed in each cell $\mycell$ of the
partition $\detpart_\ell$, with input the union of the outputs of the
unit steps applied to the children of $\mycell$. Our MST unit step
also outputs a subset of the edges of a spanning tree in addition to
the input for the next level. In particular, the unit step computes a
minimum spanning forest of the (possibly disconnected) subgraph
consisting of edges between points in the cell of length at most an
$\epsilon \levdiam_\ell$. By not including longer edges we ensure that
ignoring the edges that cross cell boundaries does not cost us a
constant factor in the quality of the approximation (see
Figure~\ref{fig:mst-bad-example}). The edges of the computed minimum
spanning forest are output as a part of the constructed spanning
tree. For the next level we output an $\epsilon^2
\levdiam_\ell$-covering of points in the cell, annotated by the
connected components of the minimum spanning forest. In a space of
constant dimension we can construct such covering of size
$\epsilon^{-O(1)}$. The reason why the distance information given by
the covering is accurate enough for our approximation is that all
edges between different connected components in the spanning forest
constructed so far are either long or have been crossing in the
previous level.

We describe the unit step as Algorithm~\ref{alg:unit-step}. Then 
Theorem~\ref{thm:mstMain} will follow from Theorem~\ref{thm:unitstep}
and the guarantees on space and time complexity, as well as the
approximation guarantees, for Algorithm~\ref{alg:unit-step}. 

\begin{algorithm}
  \caption{Unit Step at Level $\ell$}\label{alg:unit-step}  
  \SetKwInOut{Input}{input}\SetKwInOut{Output}{output} 

  \Input{Cell $\mycell \in \detpart_\ell$; a collection $V(\mycell)$
    of points in $\mycell$, and a partition $\ccomp = \{\ccomp_1,
    \ldots \ccomp_k\}$ of $V(\mycell)$ into previously computed
    connected components.}

  \Repeat{$\theta > \epsilon \levdiam_\ell$} {
    Let $\tau = \min_{\substack{i, j\\i\ne j}} \min_{u \in \ccomp_i, v \in \ccomp_j} \rho(u,v)$.\\
    Find $u \in \ccomp_i$ and $v \in \ccomp_j$ for some $i, j: i \ne
    j$ such that $\rho(u,v) \leq
    (1+\epsilon)\tau$.\\
    Let $\theta = \rho(u,v)$.\\
    \If{$\theta \le \epsilon \levdiam_\ell$}{
      Output tree edge $(u,v)$.\\
      Merge $\ccomp_i$ and $\ccomp_j$ and update $\ccomp$. \\
    } }\Output{$V' \subseteq V$, an $\epsilon^2
    \levdiam_\ell$-covering for $\mycell$, the partition $\ccomp(V')$
    induced by $\ccomp$ on $V'$. }
\end{algorithm}

Notice that Algorithm~\ref{alg:unit-step} implements a variant of
Kruskal's algorithm, with the caveats that we ignore edges longer than
$\epsilon\levdiam_\ell$ as well as edges crossing the boundary of
$\mycell$, and that we also join only the approximately closest pair
of connected components, rather than the closest pair. This last
choice is made in order to allow us to use algorithms for approximate
nearest neighbor search~\cite{Indyk-thesis,Har-PeledIM12} in order to identify
which connected components to connect, and thus achieve near-linear
total running time.

\junk{
\begin{algorithm}
  \caption{Recursive Application of the Unit
    Step}\label{alg:recursive}
    \SetKwInOut{Input}{input}\SetKwInOut{Output}{output}
  \Input{A metric space $M(S, \rho)$; a hierarchical partition
    $\detpart= (\detpart_0,\ldots, \detpart_\lev)$ sampled from an $(a,
    b, c)$-distance preserving partition, such that $\detpart_0$
    is a partition into singletons.}
  
  $\tour_0 = \emptyset$\\
  \ForEach{$\mycell \in \detpart_1$}{
    Apply the unit step (Algorithm~\ref{alg:unit-step}) to all points
    in $\mycell$ with $\ccomp$ a partition into singletons.\\
  }
  \For{levels $\ell = 2, \ldots, \lev$}{
    \ForEach{$\mycell \in \detpart_\ell$}{
      Let $\mycell_1, \ldots, \mycell_c$ be the children of
      $\mycell$\\
      For each $i \in [c]$, let $V'(\mycell_i)$ be the output of the unit step
      (Algorithm~\ref{alg:unit-step}) applied to $\mycell_i$, with
      corresponding partition into components $\ccomp(\mycell_i)$.
      Let $V(\mycell) = \bigcup_{i = 1}^c{V(\mycell_i)}$ and $\ccomp(\mycell) =
      \bigcup_i{\ccomp(\mycell_i)}$.\\
      Apply the unit step (Algorithm~\ref{alg:unit-step}) to
      $V(\mycell)$ with partition $\ccomp(\mycell)$.\\
    }
    Let $\tour_\ell$ be the union of $\tour_{\ell-1}$ and all tree
    edges output by all executions of unit step. \\
  }
  \Output{ $\tour = \tour_L$}
\end{algorithm}
}

Let $\topt$ be some optimum minimum spanning tree.  For a tree $T$,
let $\cost{T}$ denote the cost of the tree $\cost{ T} = \sum_{(u,v)
  \in T} \rho(u,v)$. The following theorem is our main approximation
result for the SAS algorithm with unit step Algorithm~\ref{alg:unit-step}.

\begin{theorem}\label{thm:mst}
  Let $\randpart$ be a randomized $(a,b)$-partition of $M(S, \rho)$
  with $L$ levels.  If $a \leq \frac{1}{2}$ and $\epsilon \leq
  \frac{1}{4}$ then the spanning tree $\tour$ output by the
  Solve-and-Sketch algorithm with partition $\detpart$ sampled from
  $\randpart$ and unit step Algorithm~\ref{alg:unit-step} satisfies:
  \begin{align*}
  \E_{\detpart \sim \randpart} \left[\cost{\tour}\right] \le (1 +
  \epsilon O(\lev b) )\cost{\topt}. 
  \end{align*}
\end{theorem}

It is natural to attempt to prove Theorem~\ref{thm:mst} by relating
the SAS algorithm with unit step Algorithm~\ref{alg:unit-step} to a
known MST algorithm, e.g.~Kruskal's algorithm (which our algorithm
most closely resembles). There are several difficulties, arising from
approximations that we use in order to achieve efficiency in terms of
communication, running time, and space. For example, our algorithm
only keeps progressively coarser coverings of the input between
phases, and thus does not have exact information about distances
between connected components. Nevertheless, it is known that an
approximate implementation of Kruskal's algorithm still outputs an
approximate MST\cite[Section 3.3.1]{Indyk-thesis}. In particular, an
algorithm that keeps a spanning forest, initially the empty graph, and
at each time step connects any two connected components of the current
forest that are at most a factor of $1+\epsilon$ further apart than
the closest pair of connected components, computes a spanning tree of
cost at most a factor $1+\epsilon$ larger than the cost of the
MST. However, our setting presents a further difficulty: because we
work in a parallel environment, Algorithm~\ref{alg:unit-step}
completely ignores any edges crossing the boundary of the cell it is
currently applied to. Such edges could have small length, which makes
it generally impossible to show that our algorithm implements
Kruskal's algorithm even approximately for the complete graph with
edge weights given by the metric $\rho$. Instead, we are able to
relate our algorithm to a run of Kruskal's algorithm on the complete
graph with \emph{modified edge weights} $\weight:S \times S\rightarrow
\R_+$. These weights are a function of the hierarchical partition
$\detpart$; they are always an upper bound on the metric $\rho$, and
give larger weight to edges that cross the boundaries of
$\detpart_\ell$ for larger $\ell$ (see
Definition~\ref{def:weights}). We are able to show
(Lemma~\ref{lem:weights-apx}) that the length (under $\rho$) of the
$i$-th edge output by (a sequential simulation) of our algorithm is at
most a factor $1+\epsilon$ larger than the weight (under $\weight$) of
the $i$-the edge output by Kruskal's algorithm, when run on the
complete graph with edge weights $\weight$. The proof is then
completed by arguing that for each $u, v \in S$, $\weight(u,v)$
approximates $\rho(u,v)$ in expectation when $\detpart$ is sampled
from a distance preserving partition (Lemma~\ref{lem:embedding}).


\junk{Our strategy to prove Theorem~\ref{thm:mst} is to relate the cost
$\tour$ to the cost of an MST for the complete graph with edge weights
$\weight:S \times S \rightarrow \mathbb{R}$, which are a function of
of the hierarchical partition $\detpart$. We consider the sequence of
edges $(e_1, \ldots, e_{n-1})$ output by a sequential simulation of
the SAS algorithm with unit step Algorithm~\ref{alg:unit-step}. We
compare this sequence with the sequence of edges $(e^*_1, \ldots,
e^*_{n-1})$ output by Kruskal's MST algorithm when run on the complete
graph with edge weights given by $\weight$. We show that $\rho(e_i) 
\leq (1+O(\epsilon))\weight(e^*_i)$ for each $i$, and, therefore, the
cost of $\tour$ under $\rho$ is not much larger than the cost of the
MST for the weights $\weight$. This suffices to prove
Theorem~\ref{thm:mst}, since we choose the weights $\weight$ so that,
when $\detpart$ is sampled from a distance-preserving partition,
$\weight(u,v)$ approximates $\rho(u,v)$ in expectation for each $u$
and $v$. }

We define the following types of edges based on the position of their
endpoints with respect to the space partition used by the algorithm.

\newcommand{\lc}{{\ell_{c}}}
\newcommand{\lna}{{\ell_{na}}}
\newcommand{\proc}{P^*}

\begin{definition}[Crossing and non-crossing edges]
An edge $(u,v)$ is \emph{crossing} in level $\ell$ if $\mycell_\ell(u) \neq \mycell_\ell(v)$ and \emph{non-crossing} otherwise.
\end{definition}

Also for each edge we define the \emph{crossing} level, which will be useful in the analysis:
\begin{definition}[Crossing level]\label{def:crossing-level}
  For an edge $(u,v)$ let its \textit{crossing level} $\lc(u,v)$ be
  the largest integer such that $\mycell_{\lc (u,v)}(u) \neq
  \mycell_{\lc(u,v)}(v)$. 
\end{definition}

The modified weights $w$, which we use in the analysis, are defined for each pair $(u,v)$ using its crossing level as follows:

\begin{definition}[Modified weights]\label{def:weights}
Let $\randpart$ be a randomized $(a, b)$-partition of $M(S, \rho)$
  with $L$ levels. 
  For every deterministic partition $P$ in the support of $\randpart$
  we define $\weight(u, v) =
  \rho(u, v) + \epsilon \levdiam_{\lc(u,v)}$. 
\end{definition}

We show that the modified weights $\weight(u,v)$ approximate the
original distances $\rho(u,v)$ in expected value. This lemma and its
proof are similar to arguments used in recent work on approximating
the Earth-Mover Distance in near-linear time~\cite{SA12-emd}, and
dating back to Arora's work on approximation algorithms for the
Euclidean Traveling Salesman Problem~\cite{A98-TSP}.

\begin{lemma}\label{lem:embedding}
  For all $u, v \in S$: 
  \begin{equation*}
    \rho(u,v) \le \E_{\detpart \sim \randpart}\left[ \weight(u, v) \right]\leq
    (1+\epsilon \lev b)\rho(u, v).
  \end{equation*}
\end{lemma}
\begin{proof}
  The lower bound $\E_{\detpart \sim \randpart}\left[ \weight(u,
    v)\right] \ge \rho(u,v)$ follows from Definition~\ref{def:weights} since $\levdiam_{\lc(u,v)} \geq
  0$. For the upper bound we use the properties of $(a,
  b)$-partitions.  We have:
  \begin{align*}
    &\E_{\detpart \sim \randpart}\left[ \weight(u, v)\right] = \E_{\detpart \sim \randpart}\left[ \rho(u, v) + \epsilon \levdiam_{\lc(u,v)} \right]\\
    &= \rho(u, v) + \sum_{\ell =
      1}^{\lev}{\Pr_{\detpart \sim \randpart}[\lc(u, v) = t]  \epsilon \levdiam_\ell}\\
      & = \rho(u,v) + \sum_{\ell = 1}^{\lev} \Pr_{\detpart \sim \randpart}\left[\mycell_{\ell}(u) \neq \mycell_{\ell}(v),  \mycell_{\ell + 1}(u) = \mycell_{\ell + 1}(v)\right] \epsilon \levdiam_\ell\\
    &\leq \rho(u, v) + \sum_{\ell = 1}^\lev{\Pr_{\detpart \sim \randpart}[\mycell_\ell(u) \neq \mycell_\ell(v)] \epsilon\levdiam_\ell} \\
     &\leq \rho(u, v) + \sum_{\ell = 1}^\lev{\epsilon b \rho(u,v)} \\
    &= \rho(u, v)(1 + \epsilon \lev b),
  \end{align*}
  where the first equality follows from Definition~\ref{def:weights}, the second equality is an expansion of the expectation, the third equality follows from Definition~\ref{def:crossing-level}, the fourth inequality follows from a term-by-term estimation of the probability of a joint event by the probability of one of its sub-events, the fifth inequality follows from the bound on the probability of cutting an edge for an $(a,b)$-partition (Definition~\ref{def:partition}) and the last one is a direct calculation. 
\end{proof}

Recall that in Algorithm~\ref{alg:unit-step} for a cell $C \in P_\ell$ the set $V(C)$ is a subset of points of $C$ considered at level $\ell$. 
Also recall that $C_\ell(u)$ is the cell containing $u$ at level $\ell$.
We use the following notation to denote the closest neighbor of $u$ considered at level $\ell$.
\begin{definition} [Nearest neighbor at level $\ell$]
  For $u \in S$ let $N_\ell(u)$ be the nearest neighbor to $u$ in
  $V(\mycell_\ell(u)) \cap \mycell_{\ell-1}(u)$, i.e. $N_\ell(u) = \arg \min_{v \in
    V(\mycell_\ell(u))\cap \mycell_{\ell-1}(u)} \rho(u,v)$.
\end{definition}

For two points $u,v$ we will use the following distance measure $\rho_\ell(u,v)$ in the analysis.

\begin{definition}[Distance between nearest neighbors at level $\ell$]\label{def:nn-distance}
 For an edge $(u,v)$ we define $\rho_\ell(u,v) = \rho(N_\ell(u), N_\ell(v))$ to be the distance between the nearest neighbors of $u$ and $v$ at level $\ell$.
\end{definition}

 The next lemma shows that  $\rho_\ell$ is an approximation to $\rho$. 
\begin{lemma}\label{lem:net}
  For every $C \in P_\ell$  and $u, v \in \mycell$ it holds that $|\rho_\ell(u, v)
  - \rho(u, v)| \le 2 \epsilon^2\levdiam_{\ell - 1}$.
\end{lemma}
\begin{proof}
  For each $\mycell' \in \detpart_{\ell - 1}$ such that $\mycell'
  \subseteq \mycell$ by construction of the algorithm it holds that $V(C)
  \cap \mycell'$ is an $\epsilon^2 \levdiam_{\ell - 1}$-covering for
  $\mycell'$.  Thus, $V(C)$ is an $\epsilon^2 \levdiam_{\ell -
    1}$-covering for $\mycell$.  We have $\rho(N_\ell(u),u) \le
  \epsilon^2 \levdiam_{\ell - 1}$ and $\rho(N_\ell(v),v) \le 
  \epsilon^2 \levdiam_{\ell - 1}$.  By the triangle inequality we get
  $\rho(N_\ell(u),N_\ell(v)) \le \rho(N_\ell(u),u) + \rho(u,v) +
  \rho(N_\ell(v),v) \le \rho(u, v) + 2 \epsilon^2 \levdiam_{\ell -
    1}$.  By another application of triangle inequality we have
  $\rho(N_\ell(u),N_\ell(v)) \ge \rho(u,v) - \rho(N_\ell(v),v) -
  \rho(N_\ell(u),u) \ge \rho(u, v) - 2 \epsilon^2 \levdiam_{\ell -
    1}$.
\end{proof}

To complete our analysis we need to further characterize edges
according to their status during the execution of the algorithm.
\begin{definition}[Short and long edges]\label{def:short-long}
  An edge $(u, v)$ is \emph{short} in level $\ell$ if $\rho_\ell(u,v)
  \le \frac{\eps}{1+\eps} \levdiam_\ell$, and long otherwise.
\end{definition}  

\begin{definition}[Processing level and sequence]\label{defn:proc-seq}
  For an edge $e$ in $\tour$, the \emph{processing level} $\ell_p(e)$
  is the integer $\ell$, such that $e$ is output by the unit step
  applied to some $\mycell \in \detpart_\ell$. Consider a sequential
  simulation of the SAS algorithm with unit step
  Algorithm~\ref{alg:unit-step}, in which at each level $\ell$, the
  unit step is applied to each cell $\mycell \in \detpart_\ell$
  sequentially in an arbitrary order. The \emph{processing sequence}
  $(e_1, \ldots, e_{n-1})$ consists of the edges of $\tour$ in the
  order in which they are output by the above sequential simulation.
\end{definition}

\newcommand{\inter}[1]{I_{#1}}

\begin{definition}[Intercluster edges]
 The
  \emph{forest at step $i$}, denoted $\tour_i$, is defined as the
  forest $\{e_1, \ldots, e_i\}$. An edge $e=(u,v)$ is
  \emph{intercluster at step $i$} if $u$ and $v$ lie in different
  connected components of $\tour_{i-1}$. We denote the set of all intercluster edges at step i as $\inter{i}$.
\end{definition}

\begin{lemma}\label{lem:ccomp}
  Let $\epsilon \le \frac{1}{4}$ and $a \le \frac{1}{2}$. For every vertex $u$, level $\ell$ and step $i > 1$ such that $\ell_p(e_i) \geq
    \ell$ the vertices $u$ and $N_{\ell}(u)$ are in the same connected
  component of $\tour_{i-1}$.
\end{lemma}
\begin{proof}
  Note that it suffices to prove the claim for the smallest $i$ such
  that $\ell_p(e_i) \ge \ell$.  Fix such $i$.  Assume for
  contradiction that for some $\ell$ the vertices $u$ and
  $N_{\ell}(u)$ are in different connected components of
  $\tour_{i-1}$. Fix the smallest such $\ell$. If $\ell = 1$, then $N_\ell(u) = u$, so we may assume $\ell
  \ge 2$. Let $\mycell = \mycell_{\ell}(u)$ and $\mycell' =
  \mycell_{\ell-1}(u)$. At the end of the execution of
  Algorithm~\ref{alg:unit-step} in cell $\mycell'$, the partition
  $\ccomp$ of $V(\mycell')$ into connected components is a subdivision
  of the connected components of $\tour_{i-1}$ restricted to
  $V(\mycell')$. By the choice of $\ell$, $u$ and $N_{\ell-1}(u)$ are
  in the same connected component of $\tour_{i-1}$, and, since we
  assumed that $u$ and $N_\ell(u)$ are in different connected
  components of $\tour_{i-1}$, it must be the case that
  $N_{\ell-1}(u)$ and $N_{\ell}(u)$ are in different connected
  components in $\ccomp$, i.e.~$N_{\ell-1}(u) \in \ccomp_k$ and
  $N_{\ell}(u) \in \ccomp_{k'}$ for $k \neq k'$. Since $V(\mycell)
  \cap \mycell'$ is a $\epsilon^2\levdiam_{\ell-1}$-covering of
  $\mycell'$ and $V(\mycell')$ is a
  $\epsilon^2\levdiam_{\ell-2}$-covering of $\mycell'$, we have
  $\rho(u, N_{\ell}(u)) \le \epsilon^2\levdiam_{\ell-1}$ and $\rho(u,
  N_{\ell-1}(u)) \le \epsilon^2\levdiam_{\ell-2}$. Then, by the
  triangle inequality, $\tau \leq \rho(N_{\ell-1}(u), N_{\ell}(u)) \le
  (1 + a)\epsilon^2\levdiam_{\ell-1}$, and the algorithm will find $u'
  \in \ccomp_k, v' \in \ccomp_{k'}$ such that $\theta = \rho(u', v')
  \leq (1+\epsilon)\tau \leq
  \epsilon^2(1+\epsilon)(1+a)\levdiam_{\ell-1}$. Since for $\epsilon
  \le \frac{1}{4}$ and $a \leq \frac{1}{2}$,
  $\epsilon(1+\epsilon)(1+a) < 1$, this contradicts the termination
  condition for the main loop of Algorithm~\ref{alg:unit-step}.
\end{proof}



\newcommand{\ekru}[1]{e^{w_P}_{#1}}

Lemma~\ref{lem:weights-apx} is the key part of the proof of Theorem~\ref{thm:mst}.
It shows that the cost of the $i$-th edge output by our algorithm is bounded in terms of the cost of $i$-th edge output by Kruskal's algorithm.

\begin{definition}[Kruskal's edge at step $i$]
Let $\ekru{i}$ be the $i$-th edge output by
  Kruskal's algorithm when run on the complete graph on $S$ with edge
  weights $\weight:S\times S \rightarrow \R$.
\end{definition}

\begin{lemma}\label{lem:weights-apx}
  If $\epsilon \leq \frac{1}{4}$ and $a \leq \frac{1}{2}$, then
  for each $i$ it holds that $\rho(e_i) \leq (1+O(\epsilon)) \weight(\ekru{i}).$
\end{lemma}
\begin{proof}

\newcommand{\emin}[1]{e^+_{#1}}

We denote the shortest intercluster edge at step $i$ as
$\emin{i} = \argmin_{e \in \inter{i}} \weight(e).$

First we show that the weight of $\ekru{i}$ is bounded by the weight
of $\emin{i}$ in Proposition~\ref{prop:kruskal-seq}. This argument is
due to Indyk~\cite[Section 3.3.1, Lemma 11]{Indyk-thesis}. 

\begin{proposition}\label{prop:kruskal-seq}
  For each $i$ it holds that $\weight(\emin{i}) \le \weight(\ekru{i})$.
\end{proposition}
\begin{proof}
Note that $\tour_{i - 1}$ has $n - i + 1$ connected components.
Because $\{\ekru{1}, \dots, \ekru{i}\}$ is a forest, there exists $j \le i$ such the endpoints of $\ekru{j}$ lie in different connected components of $\tour_{i - 1}$.
Thus, by definition of $\emin{}$ we have $\weight(\emin{i}) \le \weight(\ekru{j})$.
Because the edges output by Kruskal's algorithm satisfy that $\weight(\ekru{j}) \le \weight(\ekru{i})$ for $j \le i$ the lemma follows.
\end{proof}

Using Proposition~\ref{prop:kruskal-seq} it suffices to show that $\rho(e_i) \le (1 + O(\epsilon)) \weight(\emin{i})$ to complete the proof.
 We consider three cases:
 \paragraph{Case I: $\ell_{c}(\emin{i}) \ge \ell_p(e_{i})$.} 
 In this case we have: 
\begin{align*}
  \rho(e_i) \leq \epsilon\levdiam_{\ell_p(e_{i})} \le \rho(\emin{i}) +
  \epsilon\levdiam_{\ell_p(e_{i})} \le \rho(\emin{i}) + \epsilon \levdiam_{\lc(\emin{i})} = \weight(\emin{i}),
\end{align*}
where the first inequality follows from the condition for outputting the edges in Algorithm~\ref{alg:unit-step},
the second one is since $\rho(\emin{i}) \ge 0$,
the third one is because $\ell_p(e_{i}) \le \lc(\emin{i})$ by assumption and the last one is by Definition~\ref{def:weights}. This completes the analysis of the first case.

The following proposition will be crucial for the analysis of the remaining two cases. It shows that the $i$-th  edge $e_i$ output by (the
sequential simulation) of the SAS algorithm is approximately the
shortest non-crossing intercluster edge. 
\begin{proposition}\label{prop:min-edge}
  Let $\epsilon \le \frac{1}{4}$ and $a \le \frac{1}{2}$. If $e \in I_i$ is
  non-crossing at level $\ell_p(e_i)$ then
  $\rho(e_i) \leq (1+\epsilon)\rho_{\ell_p(e_i)}(e)$.
\end{proposition}
\begin{proof}
Fix $e = (u,v)$ and consider an edge $(u', v')$ where $u' = N_{\ell_p(e_i)} (u)$ and $v' = N_{\ell_p(e_i)} (v)$.
By Lemma~\ref{lem:ccomp} we have that $u$ and $u'$ are in the same connected component of $\tour_{i - 1}$. Similarly $v$ and $v'$ are also in the same connected component.
By assumption $(u,v) \in I_i$ so these two connected components are different. Because the edge $(u,v)$ is non-crossing at level $\ell_p(e_i)$ the edge $(u', v')$ is also non-crossing at this level. Thus,
\begin{align*}
\rho(e_i) \le (1 + \epsilon) \tau \le (1 + \epsilon) \rho(N_{\ell_p(e_{i})}(u), N_{\ell_p(e_{i})}(v)) = (1 + \epsilon)\rho_{\ell_p(e_i)}(e),
\end{align*}
where the first inequality is by construction used in Algorithm~\ref{alg:unit-step}, the second is by definition of $\tau$ together with the fact that $(u',v')$ is a non-crossing edge at level $\ell_p(e_i)$ between two different connected components and the last is by Definition~\ref{def:nn-distance}.
\end{proof}

  \paragraph{Case II: $\ell_c(\emin{i}) = \ell_p(e_{i})-1$.} 
 In this case we have:
  \begin{align*}
    \rho_{\ell_p(e_{i})}(\emin{i}) &\leq \rho(\emin{i}) + 2\epsilon^2\levdiam_{\ell_p(e_{i})-1}  \\
    &< \rho(\emin{i})  + \epsilon \levdiam_{\ell_p(e_{i})-1} + \epsilon^2\levdiam_{\ell_p(e_{i})-1} \\
    &\le (1 + \epsilon)(\rho(\emin{i}) + \epsilon \levdiam_{\ell_p(e_{i}) - 1}) \\
    &= (1 + \epsilon) \weight(\emin{i}),
  \end{align*}
  where the first line follows by Lemma~\ref{lem:net}, the second uses
  the assumption that $\epsilon < 1$, the third follows
  since $\rho(\emin{i}) \ge 0$ and the last one holds by
  Definition~\ref{def:weights} together with the assumption that
  $\ell_c(\emin{i}) = \ell_p(e_{i}) - 1$.  By assumption $\emin{i}$ is non-crossing at
  level $\ell_p(e_{i})$, and therefore Proposition~\ref{prop:min-edge}
  and the assumption $\epsilon \le \frac{1}{2}$ imply 
  \[
  \rho(e_i) \leq (1+\epsilon)\rho_{\ell_p(e_{i})}(\emin{i}) \leq
  (1+\epsilon)^2\weight(\emin{i}) \leq (1+\frac{5}{2}\epsilon)\weight(\emin{i}).
  \]

  \paragraph{Case III: $\ell_c(\emin{i}) < \ell_p(e_{i}) - 1$.}
  
  First, we prove the following auxiliary statement.
  
  \begin{proposition}\label{prop:active} 
    Let $\epsilon \le \frac{1}{4}$ and $a \le \frac{1}{2}$. 
  Every $e \in \inter{i}$ is either crossing or
    long in level $\ell_p(e_i) - 1$.
  \end{proposition}
  \begin{proof}
    Let the edge $e=(u,v)$ be 
    short and non-crossing  in level $\ell_p(e_i)-1$. We will show that this
    implies that $u$ and $v$ are in the same connected component of
    $\tour_{i-1}$ and hence $e \notin I_i$. By Lemma~\ref{lem:ccomp}, $u$ and
      $N_{\ell_p(e_{i}) - 1}(u)$ are in the same connected component of
      $\tour_{i-1}$. The same is true for  $v$ and $N_{\ell_p(e_{i}) - 1}(v)$. Thus, it suffices to show that $N_{\ell_p(e_{i}) - 1}(u)$ and $N_{\ell_p(e_{i}) - 1}(v)$ are in the same connected component of $\tour_{i - 1}$.
    
   First, note that because the edge $(u,v)$ is non-crossing at level $\ell_p(e_i) - 1$, the edge $(N_{\ell_p(e_{i}) - 1}(u), N_{\ell_p(e_{i}) - 1})$ is also non-crossing at this level.
    Suppose for the sake of contradiction that $N_{\ell_p(e_{i}) - 1}(u)$ and $N_{\ell_p(e_{i}) - 1}(v)$ were in different connected components 
    when Algorithm~\ref{alg:unit-step} finishes processing cells at level $\ell_p(e_{i}) - 1$. 
    Then Algorithm~\ref{alg:unit-step} would have found vertices $u'$ and $v'$ in these components such that:
    \begin{align*}
    \rho(u', v') \le (1 + \epsilon) \tau \le  (1 + \epsilon) \rho(N_{\ell_p(e_{i}) -
        1}(u), N_{\ell_p(e_{i}) - 1}(v)) = (1 + \epsilon) \rho_{\ell_p(e_i) - 1}(u,v) \le \eps\levdiam_{\ell_p(e_{i}) - 1},
    \end{align*}
    where the first inequality is by construction used in Algorithm~\ref{alg:unit-step},
    the second is by definition of $\tau$ together with the fact that
    $(N_{\ell_p(e_{i}) - 1}(u), N_{\ell_p(e_{i} - 1})$ is
    non-crossing, the third equality is by
    Definition~\ref{def:nn-distance} and the fourth inequality is by
    assumption that $(u,v)$ is short and
    Definition~\ref{def:short-long}. This contradicts the termination
    condition for Algorithm~\ref{alg:unit-step}.
    
  %
  \end{proof}
  
  In this case, by Proposition~\ref{prop:active}, since $\emin{i}$ was not
  crossing in level $\ell_p(e_{i})-1$, then $\emin{i}$ must have been long. Thus,
  \begin{align*}
  \rho(\emin{i}) &\ge \rho_{\ell_p(e_{i} - 1)}(\emin{i}) -  \epsilon^2\levdiam_{\ell_p(e_{i})-2} \\
  &>  \frac{\epsilon}{1+\epsilon} \levdiam_{\ell_p(e_{i})-1} - \epsilon^2\levdiam_{\ell_p(e_{i})-2}\\
  & = \left(\frac{1}{1 + \epsilon} - a \epsilon\right) \epsilon\levdiam_{\ell_p(e_{i})-1}\\
  &> (1- (1+a)\epsilon)\epsilon\levdiam_{\ell_p(e_{i})-1},
  \end{align*}
  where the first inequality is by Lemma~\ref{lem:net}, the second inequality is by definition of a long edge at level $\ell_p(e_{i})-1$ (Definition~\ref{def:short-long}) and the third equality is because for an $(a,b)$-partition it holds that $\levdiam_{\ell_p(e_i) - 2} = a \levdiam_{\ell_p(e_i) - 1}$.
    
  Then we have:
  \begin{align*}
    \rho(e_i) &\leq (1+\epsilon)\rho_{\ell_p(e_{i})}(\emin{i})  \\ 
    &\le (1+\epsilon)\rho(\emin{i}) + (1+\epsilon)\epsilon^2\levdiam_{\ell_p(e_{i})-1} \\
    &\le \left(1 + \left(1 + \frac{1+\epsilon}{1 - (1+a)\epsilon
          }\right)\epsilon\right) \rho(\emin{i}) \\
     &\le (1 + O(\epsilon))\weight(\emin{i}),
  \end{align*}
  where the first inequality is by Proposition~\ref{prop:min-edge},
  the second is by~\ref{lem:net}, the third is using the calculation
  above and the last one is a direct calculation.
\end{proof}

\newcommand{\eour}[1]{e_{#1}}

\begin{proof}[Proof of Theorem~\ref{thm:mst}]
We have:
\begin{align*}
\E_{\detpart \sim \randpart}\left[\rho(\tour)\right]
& = \E_{\detpart \sim \randpart} \left[\sum_{i = 1}^{n - 1} \rho(\eour{i})\right] \\
&\le (1 + O(\epsilon)) \E_{\detpart \sim \randpart} \left[ \sum_{i=1}^{n - 1} \weight(\ekru{i})\right]\\
& \le (1 + O(\epsilon)) \E_{\detpart \sim \randpart} \left[ \weight(\topt)\right] \\
& \le (1 + O(\epsilon)) (1 + \epsilon L b)  \rho(\topt) \\
& = (1 + \epsilon O(L b)) \rho(\topt). 
\end{align*}
Here the first equality is by linearity of expectation. The second
inequality is by Lemma~\ref{lem:weights-apx}. The third inequality holds because $\{\ekru{i}\}_{i = 1}^{n - 1}$ is an optimum minimum spanning tree for the weights $\weight$, while $\topt$ is some minimum spanning tree. The fourth inequality
is by Lemma~\ref{lem:embedding}. This completes the proof of
Theorem~\ref{thm:mst}.

\end{proof}

\subsection{Proof of Theorem~\ref{thm:mstMain}}

Theorem \ref{thm:mstMain} follows from Theorem \ref{thm:unitstep},
Theorem~\ref{thm:mst}, and the following lemma, which gives guarantees
on the time and space complexity of Algorithm~\ref{alg:unit-step}.
\begin{lemma}\label{lm:unit-step-compl}
  When the input metric space $M$ is a subset of $\ell_2^d$, the unit
  step Algorithm~\ref{alg:unit-step} has space complexity $s_u(n_u) =
  n_u\log^{O(1)} n_u$ words, and time complexity $t_u(n_u) = \eps^{-d}n_u\log^{O(1)}
  n_u$. Moreover, when we run with the Euclidean space partition
  described in Section~\ref{sec:preliminaries}, the output size is
  $p_u = O(\eps^{-d})$ words. 
\end{lemma}

Before we prove Lemma~\ref{lm:unit-step-compl}, we need to state a
couple of useful results for approximate nearest neighbor search
algorithms. 

\begin{definition}\label{def:ccp}
  In the (dynamic) \emph{$\epsilon$-chromatic closest pair problem}
  ($\epsilon$-CCP), the input is a metric space $M=(S, \rho)$ and a
  \emph{$k$-coloring} function $f:S \rightarrow [k]$. We are required
  to maintain under insertions and deletions of points an approximate
  chromatic closest pair, i.e.~a pair $u,v$ of
  points such that $f(u) \neq f(v)$ and
  $\rho(u,v) \leq (1+\epsilon)\min_{\substack{u, v\\f(u) \neq
      f(v)}}{\rho(u,v)}$. 
\end{definition}

We also need to define the classical approximate nearest neighbor
problem. 
\begin{definition}\label{def:dyn-anns}
  In the (dynamic) \emph{$\epsilon$-approximate nearest neighbor
   search problem} ($\epsilon$-ANNS), the input is a metric space $M = (S,
  \rho)$. We are required to maintain a data structure $\mathcal{D}$
  under insertions and deletions of points, so that on query specified
  by a point $q \in S$, $\mathcal{D}$ allows us to compute  a
  point $u \in S$ such that $\rho(u, q) \leq (1+\epsilon)\min_{v \in
    S}{\rho(v, q)}$.
\end{definition}

The following general reduction was proved by
Eppstein~\cite{Eppstein95-ccp}. 
\begin{theorem}[\cite{Eppstein95-ccp}]\label{thm:ccp}
  Let $T(n)$ ($n = |S|$ is the input size) be a (monotonic, bounded by
  $O(n)$) upper bound on the time required by any operation
  (insertion, deletion, or query) for a data structure solving the
  dynamic $\epsilon$-ANNS problem, and let $S(n)$ be an upper bound on
  the space of the data structure. Then one can construct a data
  structure for the dynamic $\epsilon$-CCP problem with $O(T(n)\log n
  \log k)$ amortized time per insertion, and $O(T(n)\log^2 n \log k)$
  time per deletion. The space complexity of the $\epsilon$-CCP data
  structure is $O((n + T(n))\log n)$
\end{theorem}

In our algorithm, we use the approximate nearest neighbor data
structure for Euclidean space $\ell_2^d$ from \cite{AMNSW, eppstein2008skip}:

\begin{theorem}[\cite{AMNSW, eppstein2008skip}]\label{thm:eucl-anns}
  When $M = (S, \ell_2^d)$ is a subset of $d$-dimensional Euclidean
  space, there exists a data structure for the dynamic $\epsilon$-ANNS
  problem with $O(nd\log n)$ space and preprocessing, and
  $O(\eps^{-d}\log n)$ query and update time.
\end{theorem}

\begin{proof}[Proof of Lemma~\ref{lm:unit-step-compl}]
  The proof resembles the arguments in \cite[Section 3.3.1, Lemma
  11]{Indyk-thesis}. 

  At the start of the execution of Algorithm~\ref{alg:unit-step} in a
  cell $\mycell$, we insert all points in $V(\mycell)$ into a data
  structure $\mathcal{D}$ for the $\epsilon$-CCP problem. Moreover,
  each point $u$ is given a color corresponding to the connected
  component in the initial partition $\ccomp$ of $V(\mycell)$. Then
  for the approximate chromatic closest pair $(u, v)$, check if
  $\rho(u,v) \leq \epsilon\Delta_\ell$. Assume without loss of
  generality that the current connected component of $u$ has
  cardinality no larger than that of the current connected component
  of $v$. Then we change the color of all points in the connected
  component of $u$ to the color of the points in the connected
  component of $v$. This step can be implemented by deleting all the
  points whose color needs to be changed, and re-inserting them with
  the new color. Then we move to the next iteration of the algorithm.

  Since each the color of each connected component is always changed
  to the color of a component of size at least as large, any time a
  point changes its color the number of points of the same color at
  least doubles. Since the total number of points is $n_u$, it follows
  that each point changes color at most $\log_2 n_u$ times. By
  Theorems~\ref{thm:ccp} and~\ref{thm:eucl-anns}, each update can be
  done in $O(\eps^{-d}\log^{O(1)} n_u)$ time. Therefore, the total running
  time is $O(n_u\eps^{-d}\log^{O(1)}n_u)$. The total space complexity is
  $O(n_ud\log^{O(1)} n_u)$ by Theorems~\ref{thm:ccp}
  and~\ref{thm:eucl-anns}.

  To compute the covering $V'$, recall that the cell $\mycell \in
  \detpart_\ell$ for the partition of Euclidean space discussed in
  Section~\ref{sec:preliminaries} (and
  Section~\ref{sec:eucl-part}) is a subset of a cube of side
  length $\Delta_\ell/\sqrt{d}$. We subdivide the cube into subcubes
  of side length $\epsilon\Delta_\ell/\sqrt{d}$ and we keep a single
  arbitrary point from $V(\mycell)$ for each subcube; the resulting
  set is $V'$. The size of $V'$ is bounded by the number of subcubes,
  which is $\epsilon^{-d}$. The set $V'$ is an $\epsilon\Delta_\ell$
  covering of $\mycell$ because each subcube has diameter
  $\epsilon\Delta_\ell$. 
\end{proof}

\begin{proof}[Proof of Theorem~\ref{thm:mstMain}]
For the proof of Theorem~\ref{thm:mstMain}, we first transform the
input so that it has polynomially bounded aspect ratio, as shown in
Section~\ref{sec:aspectRatio}. Then we construct a $(s^{-\Theta(1/d)},
d, s^{\Theta(1)})$-distance preserving partition $\randpart$ using
Lemma~\ref{lm:eucl-part}. Since we modified the input so that it has
polynomially bounded aspect ratio, setting $\lev = \Theta(\log_s n)$
for $s = n^{\gamma}$ ($\gamma < 1$ is a constant) suffices to make
sure that any hierarchical partition $\detpart$ from the support of
$\randpart$ is such that $\detpart_0$ is a partition into
singletons. Then, Theorem~\ref{thm:mstMain} follows from Theorem
\ref{thm:unitstep}, Theorem~\ref{thm:mst} (with approximation
parameter set to $\epsilon \leq \delta \epsilon/Lb$ for a small enough
constant $\delta$), and the following Lemma~\ref{lm:unit-step-compl}. 
\end{proof}

\newcommand{\dg}{{ \rho_g}}
\newcommand{\net}{\mathcal{N}}
\newcommand{\numargs}{d}

\section{EMD and Transportation Problem Cost}
\label{sec:emd}


In this section we show the parallel algorithms for computing the cost
of the Earth-Mover Distance and Transportation problems.

In the Transportation problem we are given two sets of
points $A, B$ in a metric space $(M, \rho)$, and a demand function
$\psi: A \cup B \rightarrow \N$ such that $\sum_{u \in A}{\psi(u)} = \sum_{v
  \in B}{\psi(v)}$. The Transportation cost between $A$ and $B$ given demands
$\psi$ is the value of the minimum cost flow from $A$ to $B$ such that
the demands are satisfied and the costs are given by the metric
$\rho$. Formally, $\emd_\rho(A, B, \psi)$ is the value of the following
linear program in variables $x_{uv}$:
\begin{eqnarray}
  \min \sum_{u \in A, v \in B}{x_{uv}\rho(u, v)} \label{eq:tcost-obj}\\
  \text{subject to}\\
  \sum_{v \in B}{x_{uv}} = \psi(u)
  &&\forall u \in A\\
  \sum_{u \in A}{x_{uv}} = \psi(v)
  &&\forall v \in B\\
  x_{uv} \geq 0
  &&\forall u \in A, v \in B\label{eq:tcost-nonneg}
\end{eqnarray}
When for all $u \in A, v \in B$, $\psi(u) = \psi(v) = 1$, the program
above has an optimal solution which is a matching, and, therefore, its
value is just the minimum cost of a perfect bichromatic matching. In
this case $\emd_\rho(A,B,\psi)$ is the Earth-Mover Distance between $A$ and
$B$, and we denote it simply by $\emd_\rho(A,B)$.

In this paper we are concerned with the Euclidean Transportation cost
problem, in which we assume that $A$ and $B$ are sets of points in the
plain $\R^2$, and $\rho$ is the usuall Euclidean distance
$\ell_2^2$. Therefore, for the rest of the section we assume that
$\rho$ is the Euclidean distance metric, and we write $\emd(A,B,\psi)$
for $\emd_\rho(A,B,\psi$). Our results generalize to any norm in
$\R^d$ for $d = O(1)$, but we focus on the Euclidean case for
simplicity.

The main result of this section is the following theorem.
\begin{theorem}[Transportation cost problem]
\label{thm:emdGeneral}
Let $\eps > 0$, space $s \geq (\log n)^{(\eps^{-1} \log_s
  n)^{\Omega(1)}}$, and max demand be $U=n^{O(1)}$.  Then there
exists an MPC algorithm with space parameter $s$ that, on input sets
$A, B \subseteq \R^2$, $|A| + |B| = n$, and demand function $\psi: A
\cup B \rightarrow [0,U]$ such that $\sum_{u\in A} \psi(u)=\sum_{b\in
  B} \psi(v)$, runs in $(\log_s n)^{O(1)}$ rounds and outputs a $1 +
\eps$ approximation to $\emd(A,B,\psi)$. Moreover, the local running time
per machine (per round) is polynomial in $s$.
\end{theorem}

Our methods also imply a {\em near-linear time sequential} algorithm for the
Transportation cost problem, answering an open problem from~\cite{SA12-emd}.
\begin{theorem}[Near-Linear Time]\label{thm:nearLinear}
  Let $\eps > 0$ and $U=n^{O(1)}$. There exists an algorithm with
  running time $n^{1 + o_\eps(1)}$ that, on input sets $A, B \subseteq
  \R^2$, $|A| + |B| = n$, and demand function $\psi: A \cup B
  \rightarrow [0,U]$ such that $\sum_{a\in A_\psi} \psi(a)=\sum_{b\in
    B_\psi} \psi(b)$, outputs a $1 + \eps$ approximation to $\emd(A,
  B, \psi)$.
\end{theorem}

We develop the proof in four steps. First, we impose a hierarchical
partition that will also define a modified distance metric $\dg$ that
approximates $\rho$ in expectation (our only step similar to one
from \cite{SA12-emd}). The new metric has a tree structure that allows
us to develop a recursive optimality condition for the Transportation
problem.
Second, we define a generalized cost function $F$, which captures all
the local solutions of a corresponding part/node in the tree.  Third,
we develop an ``information theoretic'' parallel algorithm; the
algorithm does not run in polynomial time but obeys the space and
communication constraints of our model. Finally, we modify the
information theoretic algorithm to obtain a time-efficient algorithm.

\junk{
In our proof, we expand the sets $A_\psi$ and $B_\psi$ to sets $A$ and
$B$, by having each point $a\in A_\psi$ ($b\in B_\psi$) repeated
$\psi(a)$ ($-\psi(b)$) times. Obviously
$\emd(A,B)=\emd(A_\psi,B_\psi)$. This expansion is necessary for
\emph{analysis only} (we do not in fact increase the input size): the
distinction between $A,B$ and $A_\psi, B_\psi$ is only apparent in the
first level of the hierarchy.}

\subsection{New distance}

For the remainder of this section, we assume that $A \cup B \subseteq
[\Delta/2]^2$ (i.e. $A$ and $B$ are sets of integer points) and
$\Delta = n^{O(1)}$. This assumption is without loss of generality ---
we show an MPC algorithm that transforms any input into this form in
Section \ref{sec:aspectRatio}.

We define a new \emph{grid distance}. For this, we construct a
randomized hierarchical partition $\detpart = (\detpart_1, \ldots,
\detpart_\lev)$ of $A \cup B$ using a quad-tree with branching factor
$c$, as described in Section~\ref{sec:preliminaries}, and in more
detail in Section~\ref{sec:eucl-part}. For each level
$\ell$ of the grid, define $\Delta_l$ to be the side-length of cells
at that level: $\Delta_\ell=\Delta/(\sqrt{c})^{L-\ell}$. At level $\ell$,
we also consider a subgrid of squares of side length $\delta
\Delta_{\ell}$ imposed over $\detpart_\ell$, where $\delta = k^{-1}$
for an integer $k$, so that any square of the subgrid is entirely
contained in a square of $\detpart_\ell$. Let the set of centers of
the subgrid be denoted $\net_\ell$: we call this set the ``net at
level $\ell$''. We denote the closest point to $u$ in $\net_\ell$ by
$N_\ell(u)$.

\begin{definition}
  The \emph{grid distance} $\dg: [\Delta/2]^2 \rightarrow \mathbb{R}$
  is defined as follows. Let $u, v \in [\Delta/2]^2$, let
  $\mycell_\ell(u)$ (resp.\ $\mycell_\ell(v)$) be the unique level
  $\ell$ grid cell containing $u$ (resp.\ $v$), and let $\ell_c(u,v)$ be
  the largest integer $\ell$ so that $\mycell_\ell(u) \neq
  \mycell_\ell(v)$.\ If $\ell_c(u,v) = 0$, then $\dg(u,v) =
  \rho(u,v)$. Otherwise, $\dg(u,v)=
  \dg(u,N_{\ell_c(u,v)}(u))+\rho(N_{\ell_c(u,v)}(u),N_{\ell_c(u,v)}(v))+\dg(N_{\ell_c(u,v)}(v),v)$.
\end{definition}
While this is a recursive definition, note that
$\ell_c(u,N_{\ell_c(u,v)}(u)), \ell(v,N_{\ell_c(u,v)}(v)) \leq \ell_c(u,v) -
1$, and hence the definition is not circular.

The following lemma shows that, for approximating Transportation/EMD,
$\dg$ approximates $\rho$ to within a multiplicative factor
arbitrarily close to $1$, with constant probability.
\begin{lemma}\label{lem:distortion}
  For any two sets $A,B \subseteq [\Delta/2]^2$, and demand function
  $\psi:A \cup B \rightarrow \N$ such that $\sum_{u \in A}{\psi(u)}
  = \sum_{v \in B}{\psi(v)}$, we have that $\emd_{\dg}(A,B,\psi)$ is a
  $1+O(\delta \lev)$ approximation of $\emd(A,B, \psi)$ with
  probability 9/10. In particular,
  $$
  \emd(A,B,\psi) \le \emd_{\dg}(A,B,\psi),
  $$
  and
  $$
  \E[\emd_{\dg}(A,B,\psi)-\emd(A,B,\psi)]\le O(\delta \lev)\cdot \emd(A,B,\psi).
  $$
\end{lemma}
\begin{proof}
  For the first part, we just note that $\dg(u,v)\ge \rho(u,v)$, which
  follows by the triangle inequality, and hence
  $\emd_{\dg}(A,B,\psi)\ge \emd(A,B,\psi)$.
  
  For the second part, consider an optimal solution $x$ to
  \eqref{eq:tcost-obj}--\eqref{eq:tcost-nonneg} for the original
  (Euclidean) metric $\rho$. We shall prove that
  for any $u\in A$ and $v= B$,
  \begin{equation}\label{eq:distortion}
  \E[\dg(u,v)-\rho(u,v)]\le O(\delta\lev)\cdot \rho(u,v).
  \end{equation}
  We first observe that \eqref{eq:distortion} suffices to prove the
  second part of the lemma and the main claim.  Since $x$ is a
  feasible solution for $\emd_{\dg}(A,B,\psi)$, we have
  \begin{align*}
  \E[\emd_{\dg}(A,B,\psi)-\emd(A,B,\psi)] &\leq \E\left[\sum_{u \in A, v \in
    B}{x_{uv}(\dg(u,v) - \rho(u, v))}\right] \\
  &\le O(\delta\lev)\sum_{u \in A, v \in B}{x_{uv}\rho(u, v)} \\
  &=
  O(\delta \lev)\cdot \emd(A,B,\psi). 
  \end{align*}
  This gives the second part of the lemma, and the main claim follows
  from Markov's inequality. We proceed to prove \eqref{eq:distortion}.

  Note that for any $\ell$ the probability that $\mycell_\ell(u) \neq
  \mycell_\ell(v)$ is at most
  $2\rho(u,v)/\Delta_{\ell}$: for a proof, see Lemma~\ref{lm:eucl-part}
 Furthermore, if this happens, then the
  extra distance, i.e., $\dg(u,v)-\rho(u,v)$, is at most 
  $$
  2\sum_{k = 2}^\ell{\rho(N_k(u), N_{k-1}(u))}\leq \sum_{k =
    2}^\ell{\delta\Delta_k + \delta\Delta_{k-1}} = O(\delta)
  \Delta_\ell. 
  $$
  Hence, the expected extra distance is:
  $$
  \E[\dg(u,v)-\rho(u,v)]\le \sum_{\ell=0}^{\lev} O(\delta) \Delta_\ell\cdot 2\rho(u,v)/\Delta_\ell= O(\delta \lev) \rho(u,v).
  $$
  This completes the proof. 
\end{proof}

\subsection{Cost Function}\label{sec:emdcost}

For a given grid cell $C$ at level $\ell$, we define a multi-argument
function $F$ that will represent the cost of solutions of the cell
$C$.  For the rest of this subsection, we use $\net$ to denote
$\net_\ell \cap C$. The number of arguments of $F$ is
$\numargs=1/\delta^2-1$.  In particular, there is one argument, $x_r$,
corresponding to each point $r$ from from the grid $\net$, except
exactly one (arbitrarily chosen) called $r^*\in \net$. Sometimes we
will also have variable $x_{r^*}$ for $r^*$, but which will be
entirely determined by the values of the other arguments $x_r$ (and
points inside $C$).  Let $\net^\circ=\net\setminus \{r^*\}$ be the
\emph{restricted net for $C$}. Our function $F$ has domain
$\R^\numargs$ and range of $\R_+$.

\begin{figure}[t]
  \centering
  \includegraphics{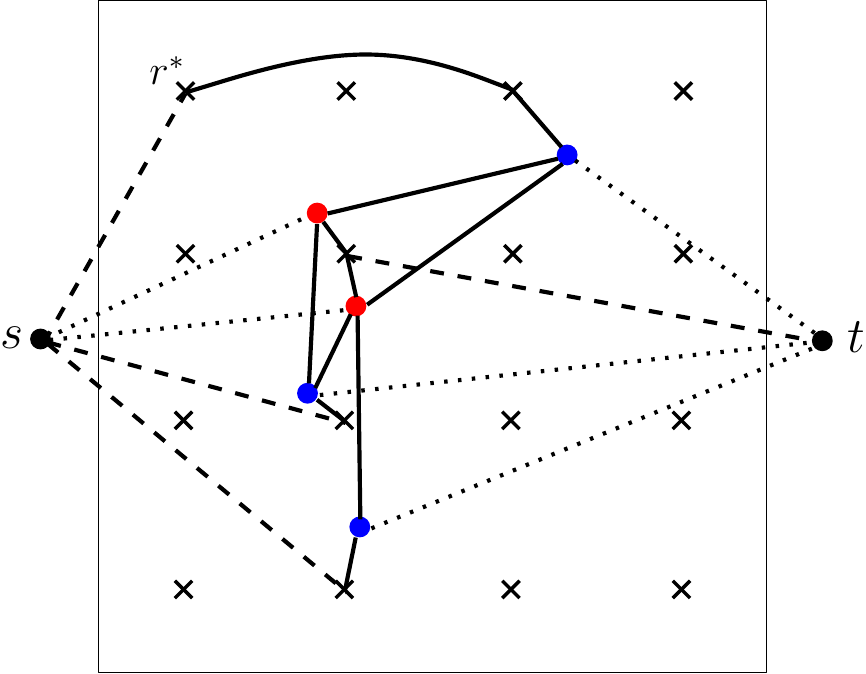}
  \caption{An example flow network for $F$. Crosses signify points in $\net$. Not all edges between $r^*$ and $\net^\circ$ are shown. Dotted edges have cost 0; dashed edges have cost $\delta\Delta_\ell/2$; solid edges have costs given by $\rho_g$. }
  \label{fig:Fcost}
\end{figure}

The function $F$ on a vector $x\in \R^\numargs$ is the value of a
solution of a min-cost flow problem. Let $A_C=A\cap C$ and $B_C=B\cap
C$ be the restriction of input to $C$. For a set of points $A'
\subseteq A$ we define $\psi(A') = \sum_{u \in A'}{\psi(u)}$, and we
define $\psi(B')$ for $B' \subseteq B$ analogously.  Define $x_{r^*} =
\psi(A_C) - \psi(B_C) - \sum_{r \in \net^\circ}{x_r}$. For each point
$r\in \net$, let $A_r = \{u \in A_C: N_\ell(u) = r\}$ and $B_r = \{v
\in B_C: N_\ell(v) = r\}$.  We set up an undirected flow network
$G_C(x)$ as follows. The nodes of the network are $s$, $t$, $\net$,
and $A_C \cup B_C$. $G_C$ includes the following edges with
corresponding costs:
\begin{itemize}
\item $s$ is connected to vertices $A_C$ with costs 0, and those
  vertices $r \in \net$ where $x_r < 0$, with costs $\delta\Delta_\ell/2$;
\item $t$ is connected to vertices $B_C$ with costs 0, and those
  vertices $r \in \net$ where $x_r > 0$, with costs $\delta\Delta_\ell/2$;
\item any two vertices $u \in A_C$ and $v \in B_C$ are
  connected, with cost $\dg(u,v)$; 
\item any $r \in \net$ is connected to $A_r \cup B_r$, and the cost of
  each edge $(r,u)$ is $\dg(r,u)$;
\item $r^*$ is connected to all $r \in \net^\circ$, with corresponding
  costs $\dg(r^*,r)$.
\end{itemize}
An example flow network is shown in Figure~\ref{fig:Fcost}. 

In the following, we use the standard convention $f(u,v) = -f(v,u)$
for any two vertices $u,v$ of $G_C(x)$. 

\begin{definition}\label{defn:F}
  The cost function $F(x)$ is defined as the cost of the minimum cost
  $s$-$t$ flow $f$ from $s$ to $t$ in $G_C(x)$, under the constraints:
  \begin{itemize}
  \item $f(s, u) = \psi(u)$ for all $u \in A_C$;
  \item $f(s,r) = -x_r$ for all $r \in \net$ for which $x_r < 0$;
  \item $f(v,t) = \psi(v)$ for all $v \in B_C$;
  \item $f(r,t) = x_r$ for all $r \in \net$ for which $x_r > 0$;
  \end{itemize}
  Above $x_{r^*} = x_{r^*} = \psi(A_C) - \psi(B_C) - \sum_{r \in
    \net^\circ}{x_r}$.  The function is defined for all $x \in
  \R^{\net^\circ}$
\end{definition}

Note that when $x_r = 0$ for all $r$, and $\psi(A_C) = \psi(B_C)$,
this flow problem corresponds exactly to the transportation cost
problem \eqref{eq:tcost-obj}--\eqref{eq:tcost-nonneg} with costs given
by $\dg$.  The interpretation of the arguments is that $x_r>0$ means
$x_r$ demand from $A_r$ has to flow to points outside $C$ through $r$
and $x_r<0$ means that $-x_r$ demand from $B_r$ has to flow in from
from points outside $C$ through $r$. Having specified all values of
$x_r$ for $r \in \net^\circ$, this leaves an imbalance of $\psi(A_C) -
\psi(B_C) - \sum_{r \in \net^\circ}{x_r}$ to flow through
$r^*$.

We first argue that the flow problem is indeed feasible for all values
of $x\in \R^{\net^\circ}$.
\begin{lemma}\label{lemma:feasiblity}
  For any $x \in \R^{\net^\circ}$ and $x_{r^*} = \psi(A_C) - \psi(B_C) -
  \sum_{r \in \net^\circ}{x_r}$, there exists a feasible flow in
  $G_C(x)$ satisfying the constraints of Definition~\ref{defn:F}.
\end{lemma}
\begin{proof}
  We construct a feasible flow $f$ as follows. We set $f(s, u)$, $f(v,
  t)$ for all $u \in A_C$, $v \in B_C$ as in
  Definition~\ref{defn:F}. For each $r \in \net^\circ$ we send flow
  $x_r$ from $r^*$ to $r$. Then we send $\min\{\psi(A_C), \psi(B_C)\}$
  units of flow from $A_C$ to $B_C$. For each $u \in A \cup B$, let
  $g(u)$ be equal to the flow going into $u$ minus the flow going
  out. Moreover, let $g(S) = \sum_{u \in S}{g(u)}$ for any $S
  \subseteq A_C \cup B_C$. For each $r \in \net$, send additional flow
  $g(u)$ from $r$ to each $u \in A_r \cup B_r$, and flow $g(A_r)$ from
  $r^*$ to $r$.
\end{proof}

The next argument is the crucial step towards sketching the function $F$. We show
that rounding each argument of $F$ to within a small enough
multiplicative factor does not change the value of $F$
significantly. The basic reason is that we can lower bound $F(x)$ in terms
of $x$, and $F$ is Lipschitz. 
\begin{lemma}
  \label{lm:round}
  Fix $\eps > 0$ and suppose $\delta=\eps^{O(1)}$. There exists an
  $\eps' = \eps^{O(1)}$ such that for all feasible $x, x' \in
  \R^{\net^\circ}$ satisfying $\forall i: x'_i/x_i \in [1 - \eps',
  1+\eps']$, we have
    $|F(x) - F(x')| \leq \eps F(x)$.
\end{lemma}
\begin{proof}
  Observe the following two properties:
  \begin{itemize}
  \item \emph{(Lower bound)} $F(x)\ge \|x\|_1 \delta \Delta_\ell/2$, because the cost of
    any edge between $s,t$ and any $r \in \net$ is
    $\delta\Delta_\ell/2$, and each point $r$ in $\net$ either sends
    flow $x_r \geq 0$ to $t$, or receives flow $-x_r \geq 0$ from $s$.
  \item \emph{(Lipschitz continuity)} for any coordinate $r \in \net$,
    we have $|F(x+\alpha e_r)-F(x)|\le O(\alpha \Delta_\ell)$, where
    $e_r$ is standard basis vector; for each feasible flow $f$ for
    $x$, we can create a flow $f'$ feasible for $x+\alpha e_r$ by
    setting $f'(r,r^*) = f(r,r^*) + \alpha$ and $f'(s, r), f'(t,r),
    f'(s, r^*), f'(t, r^*)$ are set as in Definition~\ref{defn:F};
    this flow is feasible because in $f'$ we have, by
    Definition~\ref{defn:F} 
    \begin{align*}
    f'(r^*,s) + f'(r^*, t) &= \psi(A_C) - \psi(B_C) - \sum_{r \in
      \net^\circ}{x_r} - \alpha = f(r^*,s) + f(r^*, t) - \alpha\\
    f'(r,s) + f'(r, t) &= f(r^*,s) + f(r^*, t) + \alpha
    \end{align*}
    Furthermore the difference $f - f'$ is nonzero only for edges $(r,
    r^*)$, $(r^*,s)$ and/or $(r^*,t)$, and $(r, s)$ and/or $(r,t)$;
    the cost on each of these edges is at most $\Delta_\ell/2$ and the
    change in flow along each edge is at most $\alpha$, which means
    that the cost of $f'$ is at most $O(\alpha\Delta_\ell)$ larger
    than the cost of $f$.
  \end{itemize}
  Then it follows from the above properties that
  $$ |F(x')-F(x)|\le \sum_j O(\Delta_\ell) \cdot \eps' |x_j|=O(\eps'/\delta)\cdot
  \delta \Delta_\ell/2 \|x\|_1\le \eps F(x),
  $$ as long as $\eps' < \gamma\eps\delta$ for a sufficiently small
  constant $\gamma$.
\end{proof}

\begin{lemma}
\label{lem:emdSketch}
Fix $\eps>0$ and suppose $\delta=\eps^{O(1)}$. For any nonempty cell
$C$, there is a sketch $\hat{F}$ for the function $F$ such that for
all $x$, $|F(x) - \hat{F}(x)| \leq \epsilon F(x)$, and $\hat{F}$ can
be described by a data structure of size $(\log n)^{1/\eps^{O(1)}}$.
\end{lemma}

\begin{proof}
  The sketch just remembers $F$ at all points $x$ such that each
  coordinate $x_j$ is of the form $x_j = \pm(1+\eps')^i$ for the value
  of $\eps'$ guaranteed in Lemma~\ref{lm:round} and for $i \in
  [-\log_{1+\eps'}\eps^{-O(1)}\Delta_\lev,\log_{1+\eps'}n]$ . The
  sketch also remembers the ``imbalance'', i.e.,
  $\psi(A_C)-\psi(B_C)$.  Then, to compute $F(x)$, we construct $x'$
  by rounding up each coordinate of $x$, and outputing $F(x')$. The
  approximation guarantee follows from Lemma~\ref{lm:round} if all
  coordinates $x_j$ are larger than $\eps\delta^2/\Delta_\ell$. To finish
  the proof, we claim that rounding up coordinates that are smaller
  cannot hurt the approximation factor too much. Each such coordinate
  can change the value of $F$ by at most $\eps\delta^2$, and there are
  less than $\delta^{-2}$ coordinates total, so the total change of
  value will be at most $\eps$. We claim that $F(x) \geq 1$. Indeed,
  by assumption $C$ is not empty, and the distance between any two
  points in $A_C \cup B_C$ is at least one (recall that all points in
  the input have integer coordinates). It follows that for any
  arbitrary point $u \in A_C \cup B_C$, all edges incident to $u$ have
  cost at least 1, and by the constraints of the flow defining $F$,
  the total flow outgoing (for $u \in A_C$) or incoming (for $u \in
  B_C$) is 1. Therefore, $F(x) \geq 1$, and this finishes the proof.
\end{proof}

\subsection{Information Theoretic Algorithm}

Suppose we want to compute $F$ for some cell $\mycell$ at level
$\ell$\junk{, with side-length $\Delta_\ell=\Delta/c^{(\lev-\ell)/2}$}. In
this section we give a recursive characterization of $F$ in terms of
the cost functions of the children of $C$, and we use the
characterization to approximate transportation cost with an
inefficient algorithm that still satisfies the communication
constraints of our model. In the next subsection we will modify this
algorithm so that it runs polynomial time. 

Let $\{\mycell_j\}_{j \in J}$ be the children of $C$, indexed by the
set $J$: namely, level $\ell-1$ cells contained in $C$. Let
$\{F_j\}_{j \in J}$ be the respective cost functions and $\{\hat
F_j\}_{j \in J}$ the corresponding sketches obtained from
Lemma~\ref{lem:emdSketch}; let $\tau_j$ be the imbalance of cell
$C_j$, i.e.$\tau_j = \psi(A_{C_j}) - \psi(B_{C_j})$ (stored in the
sketch). Let $\net(C_j)$ be the full net $\net_{\ell - 1} \cap C_j$ of
$C_j$ (of size $1/\delta^{2}$), and let $\net^\circ(C_j)$ be the
restricted net. Finally, let $\net$ be the net $C \cap \net_\ell$ of
$C$ and $\net^\circ = \net \setminus R^*$ the restricted version. Note
that we can pick $c$ and $\delta$ so that for any $R \in \net$, the
cluster $\{u: N_\ell(u) = R\}$ of points whose closest neighbor is $R$
is a collection of cells $C_j$, i.e. $\exists J' \subseteq J$
s.t. $\{u: N_\ell(u) = R\} \subseteq \bigcup_{j \in J'}{C_j}$. Then,
it follows that for each $C_j$ and each $u \in C_j$, $N_\ell(u) =
N_\ell(N_{\ell-1}(u))$.

We now define an optimization program used to compute $F$ on some
input $z$ using only net points in $\net$, $\net(C_j)$ for all $j \in
J$, and the functions $F_j$. For each pair $r\in \net(C_j),r'\in
\net(C_{j'})$, where $j,j'\in J$, we have a variable $x_{j,r,j',r'}$,
for all $j\neq j'$ and $r\in \net(C_j)$, $r'\in \net(C_{j'})$. The
variable $x_{j,r,j',r'}$ has the meaning that flow of value
$x_{j,r,j',r'}$ from $A_{C_j}\cup B_{C_j}$ are is routed to
$A_{C_{j'}}\cup B_{C_{j'}}$, and through $r$ and $r'$. By convention,
$x_{j,r,j',r'}=0$ whenever $j=j'$. Let $x_{j,r}=\sum_{j',r'}
x_{j,r,j',r'}$, and $x_{j, \net^\circ(C_j)}$ be the vector
$(x_{j,r})_{r \in \net^\circ(C_j)}$.

Another set of variables is $y_{j,r,R}$ for $j\in J$ and $r\in
\net(C_j)$ and $R\in \net$, where $y_{j,r,R}$ stands for the amount of
flow from $C_j$ routed via $r$ to outside $C$ via the net point
$R$. We force that $y_{j,r,R} = 0$ except when $R = N_\ell(r)$. Also
we denote $y_{j,r}=\sum_{R\in \net} y_{j,r,R}$, and
let $y_{j,\net^\circ(C_j)}$ be the vector $(y_{j,r})_{r \in
  \net^\circ(C_j)}$.

Define also the ``extra cost'' $e_{j,r,j',r'}$ to be $\rho(r,r')-
\delta\Delta_{\ell-1}$, and, similarly
$e_{j,r,R} = \rho(r,R) - \delta\Delta_{\ell-1}/2$.  Note that, because
$r,r'$ are net points at level $\ell-1$, the value of $e_{j,r,j',r'}$
and $e_{j,r,R}$ is always non-negative (for $j\neq j'$). This is the
cost that was unaccounted for by the cost functions $F_j$. 

The problem is now to minimize the objective
\begin{align}
\sum_{j} F_j(x_{j,\net^\circ(C_j)}+y_{j,\net^\circ(C_j)})+&\sum_{j,j'\in J,r\in
  \net(C_j),r'\in \net(C_{j'})} |x_{j,r,j',r'}|\cdot e_{j,r,j',r'}/2\notag\\
  + &\sum_{j\in J,r\in \net(C_j), R\in \net} |y_{j,r,R}|\cdot
  e_{j,r,R}
 + \sum_{R \in \net^\circ}{(z_R\Delta_\ell/2 + \eta_R\dg(R, R^*))}
  \label{eqn:optProblem}
\end{align}
subject to the constraints
\begin{eqnarray}
\sum_{r\in \net(C_j)}(y_{j,r}+x_{j,r})= \tau_j, 
&&
\forall j\in J
\label{eqn:cellBalance}
\\
x_{j,r,j',r'}=-x_{j',r',j,r},
&&
\forall j\in J,j'\in J,r\in \net(C_j),r'\in \net(C_{j'})
\label{eqn:antisymmetry}
\\
y_{j,r,R} = 0
&&
\forall R \neq N_\ell(r)\\
\eta_R = z_R - \sum_{j\in J,r\in \net(C_j)} y_{j,r,R}
&&
\forall R\in \net
\label{eqn:outside}
\\
z_{R^*} = \sum_{j \in J}{\tau_j} - \sum_{R \in \net^\circ}{z_R}
\label{eqn:ostracized}\\
\eta_{R^*} = -\sum_{R \in \net^\circ}{\eta_R}
\label{eqn:totalBalance}
\end{eqnarray}

\begin{lemma}\label{lm:recursive-char}
For any $z\in \R^{\net^\circ}$, the minimum value of
\eqref{eqn:optProblem} subject to constraints
\eqref{eqn:cellBalance}--\eqref{eqn:totalBalance} is equal to $F(z)$. 
\end{lemma}

\begin{proof}
Let $\alpha$ be the optimum achieved by the linear program
\eqref{eqn:optProblem}-\eqref{eqn:totalBalance}. We will prove first
that $\alpha \ge F(z)$, and then that $\alpha \le F(z)$. 

\paragraph{Clam I:  $\alpha\ge F(z)$.}

Let $x_*,y_*$ be an optimal solution to the linear program. Also, for
each $j\in J$, let $f^j$ be the flow achieving optimal cost for
$F_j(x_{j,\net^\circ(C_j)}+y_{j,\net^\circ(C_j)})$. We will construct
a flow $f$ for $F(z)$ achieving cost at most $\alpha$.

To construct $f$, we first construct a flow $f'$ in an auxiliary
network, defined as follows.  The network has vertices $s,t$, $A_C\cup
B_C$ together with $\net\cup \left(\bigcup_j \net(C_j)\right)$, and we will route
flow from $s$ to $t$. The edges in the network are defined below:
\begin{itemize}
\item $s$ is connected to $A_C$, and $t$ is connected to $B_C$; all
  these edges have cost 0;
\item $s$ is connected to vertices $R$ in $\net$ for which $z_R < 0$,
  and $t$ is connected to vertices $R$ in $\net$ for which $z_R > 0$;
  each of these edges has cost $\delta\Delta_\ell/2$
\item the network induced on $A_{C_j} \cup B_{C_j} \cup \net(C_j)$ is
  identical to the network defining $F_j$, (with $s$, $t$ and their incident
  edges removed); the costs are also identical as in the definition of
  $F_j$;
\item we have the complete graph on $\bigcup_j{\net(C_j)}$; edges between points $r \in
\net(C_j)$ and $r' \in \net(C_{j'})$ have cost $\rho(r,r')$
\item each $r \in \net(C_j)$ for some $j$ is connected to $N_\ell(r)
  \in \net$ at cost $\dg(r, N_\ell(r))$; also each $R \in \net^\circ$
  is connected to $R^*$ at cost $\dg(R,R^*)$.
\end{itemize}

The flow $f'$ between $s$, $t$, and their neighbors is defined
similarly to the flow problem defining $F$. We route flow $f'(s,u) =
1$ from $s$ to each $u \in A_C$ and flow $f(v,t) = 1$ from each $v \in
B_C$ to $t$. We also route flow $f(R,t) = z_R$ from each $R \in \net$
for which $z_R > 0$ to $t$; similarly, we route flow $f(s,R) = -z_R$
for each $R \in \net$ for which $z_R < 0$.

We define the rest of $f'$ based on $x_*,y_*, f^j$. The flow restricted
to the network induced by $A_{C_j}, B_{C_j}, \net(C_j)$ for each $j$
is exactly as in $f^j$. For each $u \in A_C$, $f'(s, u) = \psi(u)$,
and for each $v \in B_C$, $f'(v,t) = \psi(v)$. For each
$r\in \net(C_j), r'\in \net(C_{j'})$ for $j\neq j'$, we set
$f'(r,r')=x_{j,r,j',r'}$. For each $r\in \net(C_j), R\in \net$, we set
$f'(r,R)=y_{j,r,R}$. For each $R \in \net^\circ$, we set $f'(R^*,R) =
\eta_R$.

We claim that flow conservation is satisfied for $f'$. This is
immediate for each $u \in A_C$ and each $v \in B_C$, by the feasibility of
each $f^j$. By the constraints of $f^j$, each $r \in \net^\circ(C_j)$
has incoming flow $x_{j,r} + y_{j,r}$, and by the construction of $f'$
and \eqref{eqn:antisymmetry} this is also the outgoing flow. By
\eqref{eqn:cellBalance}, the total flow leaving $r_j^*$ is $\tau_j -
\sum_{r \in \net^\circ(C_j)}{(x_{j,r} + y_{j,r})}$, which, by the
constraints of $f^j$ is the incoming flow as well. Flow conservation
can be verified for $R \in \net$ using
\eqref{eqn:outside}-\eqref{eqn:totalBalance}. 

Next we claim that the cost of $f'$ is $\alpha$. For any $j$, the flow and flow
costs on edges in the network induced on $A_{C_j}$, $B_{C_j}$, $\net(C_j)$ are
exactly as in $f^j$; the total cost on these edge is
\[F_j(x_{j,\net^\circ(C_j)}+y_{j,\net^\circ(C_j)}) - \sum_{r \in
  \net(C_j)}{(x_{j,r} + y_{j,r})\delta\Delta_{\ell-1}/2}.\] The cost
of flow $f'$ along edges $(r,r')$, $r \in \net(C_j)$, $r' \in
\net(C_{j'})$ as well as that along edges $(r, R)$, $r \in \net(C_j),
R \in \net$ contribute $\sum_{r \in \net(C_j)}{(x_{j,r} +
  y_{j,r})\delta\Delta_{\ell-1}/2}$ plus the second and third term of
\eqref{eqn:optProblem}. The fourth term of \eqref{eqn:optProblem} is
given by the cost of the flow between $s$, $t$, and $\net$ together
with the cost of the flow between $R^*$ and $R \in \net^\circ$.

We are now ready to define $f$ based on $f'$. Since $f'$ is a feasible
flow, it can be decomposed into flows $f'_1, \ldots, f'_k$, where each
$f'_i$ is supported on an $s$-$t$ path. For each $u \in A$, $v \in B$,
the flow $f(u,v)$ is defined as the sum of flows $f'_i$ along paths
that pass from $s$ to $u$ to $v$ to $t$. Similarly, the flow $f(u, R)$
is defined as the sum of flows $f'_i$ along paths that pass through
both $u$ and $R$. Finally, $f(R,R^*) = f'(R,R^*)$, $f(s, R) = f'(s,
R)$, and $f(R, t) = f'(R,t)$ for each $R \in \net$. This specifies $f$.

To complete the proof of the claim we need to show that the cost of
$f$ is at most the cost of $f'$. Edges between $s$, $t$, and points in
$\net$, as well as for edges between $R^*$ and points in $\net^\circ$
exist both in the network defining $f'$ and the network defining $f$,
and they have identical flow and cost. Then, it remains to show that
for $u, v \in A_C \cup B_C$, the cost of a unit of flow in $f'$ routed
along a path that goes from $u$ to $v$ is at least $\dg(u,v)$. When
$u,v \in C_j$, this follows from the triangle inequality for $\dg$. It
then suffices to show the claim for a path consisting of edges
$(u,r)$, $(r,r')$, and $(r',v)$, where $u \in C_j$, $r =
N_{\ell-1}(u)$, $v \in C_{j'}$ ($j \neq j')$, and $r' =
N_{\ell-1}(v)$; any other path consists of such segments and paths
that are entirely within a cell $C_j$. But, by the definition of
$\dg$, since $\ell-1$ is the highest level at which $u$ and $v$ are
separated, $\dg(u,v) = \rho(r,r') + \dg(u,r) + \dg(v,r)$, and the
right hand side of this identity is exactly the cost of a unit flow of
$f'$ along the path $\{(u,r), (r,r'), (r',v)\}$. An analogous argument
shows that the cost of a unit of flow in $f'$ routed along a path that
goes from $u \in A \cup B$ to $R\in \net$ is at most $\dg(u,R)$, and,
together with the triangle inequality for $\dg$, this completes the
proof that the cost of $f$ is at most the cost of $f'$. Since the cost
of $f'$ is at most $\alpha$, and $f$ is a feasible flow satisfying the
constraints of Definition~\ref{defn:F}, it follows that $F(z) \leq
\alpha$.

\paragraph{Clam II:  $\alpha\le F(z)$.}

To prove that $\alpha \le F(z)$, we will use a flow $f$ of cost
$F(z)$, feasible for Definition~\ref{defn:F}, and construct a solution
$y_*, x_*, \eta_*$ for the LP
 \eqref{eqn:optProblem}-\eqref{eqn:totalBalance}, as well as flows
$f^j$ for $F_j$ for each $j$, so that the objective
\eqref{eqn:optProblem} is at most the cost of $f$. For any $u, v \in
C_j$, we set $f^j(u,v) = f(u,v)$. Then, for $u \in C_j$, $r =
N_{\ell-1}(u)$, and $R = N_\ell(u) = N_\ell(r)$, we set
\[
f^j(u,r) = f(u,R) + \sum_{v \in C_{j'}, j' \neq j}{f(u,v)}.
\]
This defines all $f^j$, and also $x_{j, \net^{\circ}(C_j)}$. We then
set 
\[
x_{j,r,j',r'} = \sum_{u \in A_r, v\in B_r}{f(u,v)}\hspace{3ex}\text{ and }\hspace{3ex}
y_{j,r,R} = \sum_{u \in A_r \cup B_r}{f(u,R)}.
\]
Finally, let $\eta_R =
f(R^*,R)$. The flow conservation inequalities for $f$ imply the
feasibility of $f^j$ for each $j$, as well as
\eqref{eqn:cellBalance}-\eqref{eqn:totalBalance}. 
The total of  \eqref{eqn:optProblem} and the costs of $f^j$ is equal to
the cost of $f$, by the definition of distance function $\dg$.
\end{proof}

In order to satisfy the space and communication requirements of our
model,we cannot represent the function $F_j$ exactly, and instead we
replace each $F_j$ with the sketch $\hat{F}_j$. Thus, we consider the
optimization problem with objective function
\begin{align}
\sum_{j} \hat{F}_j(x_{j,\net^\circ(C_j)}+y_{j,\net^\circ(C_j)})+&\sum_{j,j'\in J,r\in
  \net(C_j),r'\in \net(C_{j'})} |x_{j,r,j',r'}|\cdot e_{j,r,j',r'}/2\notag\\
  + &\sum_{j\in J,r\in \net(C_j), R\in \net} |y_{j,r,R}|\cdot
  e_{j,r,R}
 + \sum_{R \in \net^\circ}{(z_R\Delta_\ell/2 + \eta_R\dg(R, R^*))}
  \label{eqn:optSketchProblem}
\end{align}
subject to \eqref{eqn:cellBalance}-\eqref{eqn:totalBalance}, where
$\hat{F}_j$ is a sketch of $F_j$.

Next we define an inefficient unit step that satisfies the space and
communication constraints of the MPC model. In the following
subsection we also give a polynomial time variant of the unit step. 

\paragraph{Unit Step (information theoretic).} The (inefficient) unit
step $\alg_{i}$, when executed in cell $C$ at level $\ell$, accepts as
input all the sketches $\hat{F}_j$ for all children $C_j$ of $C$ (if
$\ell = 1$, we assume we have $F_j$ exactly). It solves
\eqref{eqn:optSketchProblem} on all inputs $z \in \R^{\net^\circ}$ for
which $z_i = \pm (1+\eps')^{j_i}$ where $j_i$ is an integer in
$[-\log_{1+\eps'}\eps^{-O(1)}\Delta_\lev, \log_{1+\eps'}n]$ and
$\eps'$ is as in Lemma~\ref{lem:emdSketch}. It then outputs the
computed objective values and the cell imbalance $\psi(A_C) -
\psi(B_C)$ as the sketch $\hat{F}$ for $F$.

\begin{theorem}\label{thm:emdinfo}
  The unit step $\alg_i$ has space complexity $s_u(n_u) = O(n_u^2)$
  and output size $p_u = (\log n)^{\eps^{-O(1)}}$. Moreover, the
  Solve-and-Sketch algorithm with unit step $\alg_i$ that outputs
  \eqref{eqn:optSketchProblem} with $z = (0, \ldots, 0)$ in the unique
  cell $C \in \detpart_\lev$ containing $A \cup B$, computes an $1 \pm
  O(\lev\eps)$ approximation to $\emd(A,B,\psi)$.
\end{theorem}
\begin{proof}
  Observe that $F(0,\ldots, 0) = \emd_{\rho_g}(A,B,\psi)$. By
  Lemma~\ref{lem:distortion}, it is then enough to prove that for all
  cells $C$ on level $\ell$, and for each $x \in \R^{\net^\circ}$, the
  value computed by \eqref{eqn:optSketchProblem} is a $(1\pm
  O(\ell\eps))$ approximation of $F(x)$. For $\ell = 1$, this is
  immediate because we compute the exact value of $F$. For $\ell>1$
  the claim follows by induction from Lemmas~\ref{lm:recursive-char},
  \ref{lem:emdSketch}, because each term of \eqref{eqn:optProblem} is
  non-negative.
\end{proof}

\subsection{Computationally Efficient Algorithm}
\label{sec:transportation_efficient_cost}

The goal in this section is to show that (a variant of) the unit step
from the previous section can be implemented in polynomial time. Then
the existence of an efficient MPC algorithm for approximating
transportation cost will follow from Theorem~\ref{thm:emdinfo}.

As before, we fix a cell $C$ on level $\ell$ and focus on
approximating $F$ evaluated in the cell. We first observe that the
function $F$ is convex in its parameters. Then we use this fact to
show that a ``convexification'' of the sketch $\hat{F}$ is also an
accurate approximation to $F$. The two lemmas follow.

\begin{lemma}\label{lm:convex}
For any cell $C$ the associated function $F$ is convex.
\end{lemma}
\begin{proof}
  Consider any $x,y\in \R^\numargs$. Consider $F(x)$ and suppose $f^x$ is the
  minimum cost flow for $F(x)$. Similarly, suppose $f^y$ is the
  minimum cost flow for $F(y)$. Then, note that $\frac{f^x+f^y}{2}$ is
  a feasible solution to $F(\frac{x+y}{2})$, and achieves a value of
  $\frac{F(x)+F(y)}{2}$ (by linearity). Since there could be other,
  lower cost solutions to $F(\frac{x+y}{2})$, we conclude that
  $F(\frac{x+y}{2})\le \frac{F(x)+F(y)}{2}$.
\end{proof}

\begin{lemma}\label{lm:convex-apx}
  Let $\eps'$ be as in Lemma~\ref{lem:emdSketch}. Let $X$ be the
  matrix whose columns are all points $x \in [-n, n]^\numargs$ such that for
  all $i \in [d]$, $x_i = \pm (1+\eps')^{j_i}$ for some integer $j_i
  \in [-\log_{1+\eps'}\eps^{-O(1)}d\Delta_\lev,\log_{1+\eps'}n]$. Let
  $f$ be the vector defined by $f_i = F(x^i)$, where $x^i$ is the
  $i$-th column of $X$. Then the function $\tilde{F}$ defined by
  \begin{equation*}
    \tilde{F}(x) = \min\{\sum_i \alpha_i f_i: x = X \alpha,\ \alpha
    \geq 0,\ \sum_i \alpha_i = 1\},
  \end{equation*}
  satisfies
  \begin{equation*}
    |F(x) - \tilde{F}(x)| \leq \eps F(x). 
  \end{equation*}
\end{lemma}
\begin{proof}
  We first use the convexity of $F$ to prove that $\tilde{F}(x) \geq
  F(x)$. For this, it is enough to observe that for any $\alpha \geq
  0$, such that $\sum_i \alpha_i = 1$, we have $F(X\alpha) \leq
  \sum_i{\alpha_i F(x^i)}$ by Lemma~\ref{lm:convex}, and in particular this holds for the
  minimizer $\alpha$ that gives $\tilde{F}(x)$. 

  Next we show that $\tilde{F}(x) \leq (1+\eps)F(x)$. Fix some $x$,
  and let, for each $j \in [d]$, $k_j$ be $k_j = \lfloor
  \log_{1+\eps'}|x_j| \rfloor$, and $s_j$ be $1$ if $x_j > 0$, and
  $-1$ otherwise. Then $x$ is in the convex hull of the points
  $S(x) = \{(s_1(1+\eps')^{k_1 + b_1}, \ldots, s_d(1+\eps')^{k_d + b_d}):
  \forall j: b_j \in \{0,1\}\}$. Therefore, $\tilde{F}(x) \leq
  \max\{F(x'): x' \in S(x)\}$. But, by Lemma~\ref{lm:round},
  $\max\{F(x'): x' \in S(x)\} \leq (1+\eps)F(x)$, and the claim
  follows. 
\end{proof}

The crucial property of $\tilde{F}$ is that it is specified as a
solution to a linear program, and as such it's easy to ``embed''
inside a larger linear program. We do this for the program
\eqref{eqn:optProblem}-\eqref{eqn:totalBalance} next.

Let $X$ be as in Lemma~\ref{lm:convex-apx} and $D$ be the number of
columns of $X$. Observe that $D = (\log n)^{\eps^{-O(1)}}$. Let, for
each $j \in J$, $f^j$ be a vector of dimension $D$, such that  $f^j_i
= F_j(x^i)$, for $x^i$ the $i$-th column of $X$. We consider the
linear program whose objective is to minimize:
\begin{align}
\sum_{j} \langle f^j, \alpha^j \rangle+&\sum_{j,j'\in J,r\in
  \net(C_j),r'\in \net(C_{j'})} |x_{j,r,j',r'}|\cdot e_{j,r,j',r'}/2\notag\\
  + &\sum_{j\in J,r\in \net(C_j), R\in \net} |y_{j,r,R}|\cdot
  e_{j,r,R}
 + \sum_{R \in \net^\circ}{(z_R\Delta_\ell/2 + \eta_R\dg(R, R^*))}
  \label{eqn:optConvProblem}
\end{align}
subject to constraints
\eqref{eqn:cellBalance}-\eqref{eqn:totalBalance} as well as the
additional constraints:
\begin{eqnarray}
  x_{j,\net^\circ(C_j)}+y_{j,\net^\circ(C_j)} = X\alpha^j
  &&\forall j \in J\label{eq:convexcomb}
\end{eqnarray}

\paragraph{Unit Step (efficient).} The (efficient) unit step
$\alg_{e}$, when executed in a cell $C$ (with net $\net$ inside the
cell, and restricted net $\net^\circ = \net \setminus \{R^*\}$) on
level $\ell > 1$, takes as input the vectors $f^j$ for all child cells
$C_j$, as well as the imbalances $\tau_j$. It solves
\eqref{eqn:optConvProblem} on all inputs $z \in \R^{\net^\circ}$ for
which $z_i = \pm (1+\eps')^{j_i}$ where $j_i$ is an integer in $
[-\log_{1+\eps'}\eps^{-O(1)}\Delta_\lev,\log_{1+\eps'}n]$ and $\eps'$
is as in Lemma~\ref{lem:emdSketch} and outputs the computed objective
value for each input to form a vector $f$. It also outputs the cell
imbalance $\psi(A_C) - \psi(B_C)$. At level $\ell = 1$, the function $F$ is
evaluated exactly.

\begin{theorem}\label{thm:emdefficient}
  The unit step $\alg_{eu}$ has polynomial space and time complexity
  (i.e. $s_u(n_u), t_u(n_u) = n_u^{O(1)}$) and output size $p_u =
  (\log n)^{\eps^{-O(1)}}$. Moreover, the solve-and-sketch algorithm
  with unit step $\alg_i$ that outputs \eqref{eqn:optConvProblem} with
  $z = (0, \ldots, 0)$ in the unique cell $C \in \detpart_\lev$
  containing $A \cup B$, computes an $1 \pm O(\lev\eps)$ approximation
  to $\emd(A,B,\psi)$.
\end{theorem}
\begin{proof}
  The proof of approximation is analogous to the proof of
  Theorem~\ref{thm:emdinfo}, but using Lemma~\ref{lm:convex-apx}
  rather than Lemma~\ref{lem:emdSketch}. The space and time complexity
  claims follow since\eqref{eqn:optConvProblem} is a equivalent to a
  linear optimization problem and can be solved in polynomial time and space.
\end{proof}

We now deduce Theorem~\ref{thm:emdGeneral}, from
Theorem~\ref{thm:emdefficient} (with $\eps'=O(\eps/L)=O(\eps/\log_s
n)$) and Theorem~\ref{thm:unitstep}.

\subsection{Consequences}

Theorem~\ref{thm:emdefficient} has consequences beyond the MPC model:
it implies a near linear time algorithm and an algorithm in the
streaming with sorting routine model for EMD and the transportation
problem. 

The first theorem follows directly from the simulation of MPC
algorithms by algorithms in the streaming with sorting routine model. 
\begin{theorem}\label{thm:emdStreamSort}
  Let $\eps > 0$, and $s \geq (\log n)^{(\eps^{-1} \log_s
    n)^{O(1)}}$. Then there a space $s$ algorithm that runs in
  $(\log_s n)^{O(1)}$ rounds in the streaming with sorting routine
  model, and, on input sets $A, B \subseteq \R^2$, $|A| + |B| = n$,
  and demand function $\psi: A \cup B \rightarrow \N$ such that
  $\sum_{u\in A} \psi(u)=\sum_{b\in B} \psi(v)$ outputs a $1 + \eps$
  approximation to $\emd(A,B,\psi)$. 
\end{theorem}

The existence of a nearly-linear time algorithm for EMD and the
transportation problem also follows easily from
Theorem~\ref{thm:emdefficient}. We stated this result as
Theorem~\ref{thm:nearLinear}, and we prove it next. 
\begin{proof}[Proof of Theorem~\ref{thm:nearLinear}]
  The algorithm simulates the Solve-And-Sketch algorithm sequentially,
  by applying the unit step $\alg_{eu}$ with approximation parameter
  $\eps' = \eps/\lev$ to each (non-empty) cell at
  level $\ell=1,\ldots,\lev$ of the partition, and finally outputs
  \eqref{eqn:optConvProblem} with $z = (0, \ldots, 0)$. We choose the branching
  factor $c$ of the hierarchical partition to be $(\log
  n)^{\omega(1)}$, so that the size of the output of $\alg_{eu}$ is
  $n^{o(1)}$, and the height $\lev$ of the partition is $\lev = o(\log
  n)$. The running time bound follows.
\end{proof}

\section{Parallel Implementation of the Solve-And-Sketch Framework}
\label{sec:implementation}
\newcommand{\netsz}{p_u}

In this section we describe how to implement in the MPC model algorithms that fit the
Solve-And-Sketch framework. In particular, we prove
Theorem \ref{thm:unitstep} from Section \ref{sec:preliminaries}. In
fact it will follow from a more general Theorem \ref{thm:generalSAS}
below. Both theorems assume the existence of a unit step algorithm
$\alg_u$.

Consider a hierarchical partition $\detpart = (\detpart_0, \ldots,
\detpart_\lev)$ of an input pointset $S$, where $\detpart_{\ell-1}$ is
a subdivision of $\detpart_{\ell}$ for each $\ell$. Recall that the
children of a cell $C \in \detpart_\ell$ are subcells $C' \in
\detpart_{\ell-1}: C' \subseteq C$, and the degree of $\detpart$ is
the maximum number of children any cell $C$ has. (For simplicity,
think of these parameters as $L=O(1)$ and $c=s^{1/3}$.) In order to
implement our algorithms, we require that each point $u$ in the input
is labeled by the sequence $(\mycell_0(u), \ldots, \mycell_\lev(u))$,
where $\mycell_\ell(u) \in \detpart_\ell$ is the unique cell on level
$\ell$ the partition that contains $u$. When the input $S$ is
represented in this way, we say that it is \emph{labeled} by
$\detpart$. In this section we discuss in detail how to label subsets
of Euclidean space; the corresponding construction for arbitrary
metric spaces of bounded doubling dimension is discussed in
Section~\ref{sec:doubling}.

We assume that we have a \emph{total ordering} on all cells with the
following property:
\begin{itemize}
\item (Hierarchical property.) For any $\ell \in \{1, \ldots, \lev\}$,
  and any two cells $C_1, C_2 \in \detpart_\ell$, $C_1 < C_2$ implies
  that for any child $C'_1$ of $C_1$ and any child $C'_2$ of $C'_2$,
  $C'_1 < C'_2$. 
\junk{\item  For two cells $C_1$ and $C_2$, if $C_2$
  is a proper subcell of $C_1$ (i.e. $C_2 \subsetneq C_1$), then $C_1
  < C_2$.
\item For two sibling cells $C_1$ and $C_2$ (i.e. children of the same
  cell $C$) such that $C_1 <
  C_2$, for all subcells $C_3$ of $C_1$, it holds $C_3 < C_2$.}
\end{itemize}
We call an ordering with the above property a \emph{good ordering}.

%

In each round of the algorithm, we sort all cells which have not yet
been processed and we assign an interval of cells to a machine. Using
the bound $\lev$ on the number of levels of the partition and the
degree bound $c$, we show that after each machine recursively applies
the unit step to each of its assigned cells, the output to the next
round of computation is significantly smaller than the input to the
current round. This enables bounding the total number of rounds by
essentially $(\log_s n)^{O(1)}$ (for good choices of $\lev$ and $c$).




%
%

We give the implementation of the Solve-And-Sketch algorithm in the
MPC model as Algorithm~\ref{alg:load-balancing}. 

\begin{algorithm}
  \caption{Implementation of algorithms in the Solve-And-Sketch framework}
  \label{alg:load-balancing}
  \SetKwInOut{Input}{input}\SetKwInOut{Output}{output} \Input{A set
    $S$, labeled by a hierarchical partition $\detpart = (\detpart_0,
    \ldots, \detpart_\lev)$ of degree $c$ such that $\detpart_0$ is a
    partition into
    singletons. 
  }

  
  \For{$r = 1, \ldots R$\nllabel{step:load-balancing:main-loop}}{

%
    Let $\mathcal{\mycell}_r$ be the be the set of non-empty cells
    $\mycell$ of $\detpart$ such that\linebreak
    \hbox{\quad\textbf{(1)}} $\mycell \in P_0$ or the unit step $\alg_u$ has already been applied to $\mycell$, and\linebreak
    \hbox{\quad\textbf{(2)}} the unit step has not been applied to the parent of  $\mycell$.\\
    Sort $\mathcal{\mycell}_r$ according to the induced order. Let the
    sorted order be $\mycell_1, \ldots,
    \mycell_{n_r}$.\\

    Let $p_u(C_i)$ be the output size of cell $C_i$. Compute
    $h_i=\sum_{j\le i} p_u(C_j)$ for all $i \leq n_r$ using a prefix sum algorithm.\\

    Consider some cell $i$, and let $j = \lceil \frac{h_i}{s} \rceil$.
    Machine $j$ receives the output of the unit step
    that has been applied in round $r-1$ to $\mycell_i$.\\

    \ForEach{\rm machine $j$}{ 
      \For{$\ell = 1, \ldots, \lev$}{ 
        \While{\rm there exists a cell $\mycell \in \detpart_\ell$
          such that $\mycell \subseteq \bigcup_{i: (j-1)s\le h_i< js}{ C_i}$}{ 

          Apply the unit step $\alg_u$ to $\mycell$, where the inputs are:\linebreak
 \hbox{\quad\bf for $\ell>1$:} the outputs of the unit step for children of $\mycell$\linebreak
 \hbox{\quad\bf for $\ell=1$:} the cells that are children of $\mycell$
\\
          (If the output of $\alg_u$ is larger than its input size, the algorithm instead just outputs the input; appropriately marked to be ``lazy evaluated''.)
        }
      }
    }
  }
\end{algorithm}

For two cells $\mycell_1
  \in \detpart_{\ell_1}$, $\mycell_2\in \detpart_{\ell_2}$, let $[\mycell_1, \mycell_2] = \{\mycell: \mycell_1 \leq \mycell
  \leq \mycell_2\}$. We call a cell $C \in [C_1, C_2]$ a
  \emph{boundary} cell if there exists a sibling $C'$ of $C$ such that
  $C' \not \in [C_1, C_2]$. 

The following lemma bounds the number of boundary cells for any
interval $[C_1, C_2]$, and is key to the analysis of
Algorithm~\ref{alg:load-balancing}. 

\begin{lemma}\label{lm:nodes-interval}
  For a partition  $\detpart = (\detpart_0, \ldots, \detpart_\lev)$
  and a range $[C_1, C_2]$, the number of boundary cells is at most $2(c-1)\lev$.
\end{lemma}
\begin{proof}
  It suffices to show that there are at most $2c - 2$ boundary cells
  in each level $\ell \in \{0, \ldots, \lev-1\}$. Indeed, suppose for
  the sake of contradiction that there are $n_\ell > 2c-2$ boundary
  cells in $P_\ell$: call the set of such sets $\mathcal{C}$, and
  number them as $C_1 \leq C^\ell_1 < C^\ell_2 < \dots <
  C^\ell_{n_{\ell}}\leq C_2$. Then there must be at least three cells
  $C_1^{\ell+1} < C_2^{\ell+1} < C_3^{\ell+1}$ ($C_1^{\ell+1},
  C_2^{\ell+1}, C_3^{\ell+1} \in \detpart_{\ell+1}$), such that each
  of them has at least one child cell in $\mathcal{C}$. But by the
  hiearchical property of good orderings, this implies that for all
  children $C$ of $C_2$, $C_1 \le C_1^\ell < C < C^\ell_{n_\ell} \le
  C_2$, and therefore none of the children of $C_2^{\ell+1}$ is a
  boundary cell, a contradiction.
\end{proof}

\begin{theorem}
\label{thm:generalSAS}
Let $\detpart$ be a hierarchical partition of the input $S$, and let $S$
be labeled by $\detpart$. Furthermore, let $\alg_u$ be a unit step
algorithm, which, for input of size $n_u$, has time
$t_u=t_u(n_u)$ and space $s_u=s_u(n_u)$, as well as parameter
$p_u=p_u(n_u)$. Assume all functions are non-decreasing, and let
$p_u=p_u(s), t_u=t_u(cp_u)$. Suppose $s_u(cp_u)\le s$, and $p_u(s)\le
\tfrac{\sqrt{s}}{4cL}$.

  Algorithm~\ref{alg:load-balancing} can be
  simulated using $R \cdot O(\log_s n)$ rounds in the MPC
  model. Furthermore, after $R=O(\log_s n)$ rounds,
  Algorithm~\ref{alg:load-balancing} has applied the unit step to all
  cells of the hierarchical partition $\detpart$. 
  If comparing two cells in some good ordering requires time $\tau$,
  then the local computation time per machine in a round is bounded by
  $O(s\cdot(\tau \cdot \log s \cdot \log_s n + Lt_u))$. All bounds are
  with high probability.
\end{theorem}

\begin{proof}
  We first show that each iteration of the loop in
  Step~\ref{step:load-balancing:main-loop} can be computed in
  $O(\log_s n)$ rounds of the MPC model with high probability. To that
  end, we use the sorting algorithm and the prefix-sum algorithm of
  Goodrich et al.~\cite{GSZ11-sorting}, which obtains this bound with
  high probability, because they operate on input of size
  $n_r = |\mathcal{\mycell}_r| \leq n$. The other steps of each iteration
  can easily be simulated in a constant number of rounds.

  Now we bound the number of rounds $R$ before the unit step has been
  applied to all cells of $\detpart$.  In round $r$, machine $j$ takes
  a range $[\mycell_{i_j}, \mycell_{i_{j+1}-1}]$ where $i_j$ are
  determined from the prefix-sum calculation.

  Let $\mathcal{\mycell}_{r+1,j}$ be the cells $\mycell$ such that
  \begin{itemize}
  \item machine $j$ applies the unit step to $\mycell$ in round $r$,
  \item the unit step is not applied to the parent of $\mycell$ in round $r$.
  \end{itemize}
  All such cells are boundary cells, and by Lemma~\ref{lm:nodes-interval},
  $|\mathcal{C}_{r+1, j}| \leq cL$; since
  $\mathcal{\mycell}_{r+1} =\bigcup_j{\mathcal{\mycell}_{r+1,j}}$, we
  have $n_{r+1} \leq cL \lceil \tfrac{n_r}{s/p_u-1 } \rceil$ (in each
  machine, there are at least $s/p_u-1$ distinct outputs that are
  processed, i.e., ``consumed''). Since $p_u\le
  \tfrac{\sqrt{s}}{4cL}$, we have that $n_{r+1} \leq n_r/\sqrt{s}$,
  i.e., there can be only $R=\log_{\sqrt{s}}n=O(\log_s n)$ rounds before
  $\mathcal{\mycell}_r$ fits entirely on a machine. Once
  $\mathcal\mycell_r$ fits on a single machine, the process finishes.

In terms of space and computation, each machine always receives at
most $O(s)$ amount of information, by construction.  Furthermore,
because of lazy evaluation, the output is always no larger than the
input for any fixed machine $j$. Hence, we also never run out of
machines when partitioning $\mycell_i$'s outputs (namely, the total
``information'' in the system remains bounded by $O(n)$, in chunks of
size at most $p_u\le s$). Finally, the space is bounded by the machine
input size, $O(s)$, plus local space required by $\alg_u$, which is
$s_u(cp_u)\le s$ since $cp_u\le s$ is the bound on the input into a
unit step.

Finally, we bound the computation time per machine. The
sorting/prefix-sum parts take time $O(s\tau_1 \cdot \log s \cdot
\log_s n)$. The rest of the computation time is dominated by the
worst-case unit step computation. Overall, no more than $O(sLt_u)$
computation per machine is necessary in a single iteration of the loop
in Step~\ref{step:load-balancing:main-loop}, once we are done with
sorting.
\end{proof}

We remark that the above theorem immediately implies Theorem \ref{thm:unitstep}.
\begin{proof}[Proof of Theorem \ref{thm:unitstep}]
Suppose $s_u(n_u), t_u(n_u)$ are all bounded by the polynomial
$(n_u)^q$, where $q \ge 1$ is constant. Fix $c=s^{1/(2q)}$. Then we get that
$L=\tfrac{\log \Delta}{c^{1/d}}=O(d\log_s n)$.

We verify we can apply Theorem \ref{thm:generalSAS}. Indeed, we have
that $p_u=p_u(s)\le s^{1/3}\le \tfrac{\sqrt{s}}{4cL}$. We also have
$s_u(cp_u(s))\le (cp_u(s))^q\le c^q\cdot p_u(s)^q\le \sqrt{s}\cdot
s^{1/3}<s$. The theorem follows.
\end{proof}


\subsection{Constructing a Partition in the Euclidean Space}
\label{sec:eucl-part}

We now describe how an $(a,b,c)$-partition can be computed for the
Euclidean space in MPC. This is an elaboration on the construction
described in Section~\ref{sec:preliminaries}, and is used for our
Euclidean MST and transportation cost algorithms. 

\begin{lemma}\label{lm:eucl-part}
  Let $a> 1$, $d$, and $L$ be positive integers. Consider a metric
  space $(S, \ell_2^d)$, where $S \subseteq \R^d$, $|S| = n$.  There
  is a $(1/a,d,(a+1)^d)$-distance-preserving partition $\randpart$ of
  $S$ with approximation $\boxapprox \leq \sqrt{d}$ and $L+1$
  levels. Moreover $S$ can be labeled by a hierchical partition
  $\detpart$ sampled from $\randpart$ in $O(\log_s n)$ rounds of MPC,
  and $\detpart$ has a good ordering of subcells such that each pair
  of cells can be compared in $O(dL)$ time.
\end{lemma}

\begin{proof}
  The partition is constructed by applying a randomly shifted grid,
  similarly to Arora's classical construction~\cite{A98-TSP}. Let us
  shift $S$ so that for all $u \in S$, and all coordinates $i \in [d]$,
  $u_i \ge 0$, and also there exists a $v \in S$ and a coordinate $i$
  such that $v_i = 0$. Let, further, $\Delta$ be the smallest real
  number such that $S \subseteq [0, \Delta]^d$. A point $r$ is
  selected uniformly at random from $[0, \Delta]^d$. Two points $u$
  and $v$ belong to the same cell at level $\ell\in\{0,\ldots,L\}$ if
  and only if for all dimensions $i \in [d]$, $\left\lfloor \frac{(u_i-r_i)
    a^{L-\ell}}{\Delta}\right\rfloor = \left\lfloor
  \frac{(v_i-r_i)a^{L-\ell}}{\Delta}\right\rfloor$.

Let us state the desired properties of the partition.
\begin{enumerate}
 \item The diameter of a cell at level $\ell$ is bounded by $\Delta_\ell =
   \sqrt{d}\Delta/a^{L-\ell}$; moreover, by the choice of $\Delta$,
   $\diam(S) = \max_{u, v \in S}{\|u - v\|_2} \geq \max_{u, v \in
     S}{\|u - v\|_\infty} = \Delta$. 

 \item Consider two points $u,v \in H$. The probability that $\left\lfloor
   \frac{(u_i-r_i)a^{L-\ell}}{\Delta}\right\rfloor \ne \left\lfloor
   \frac{(v_i-r_i)a^{L-\ell}}{\Delta}\right\rfloor$ for a given $i$ is at
   most $\frac{|u_i - v_i|a^{L-\ell}}{\Delta}$. Therefore, by the
   union bound, the probability that $u$ and $v$ belong to different
   cells is bounded by 
   \[
   \frac{\|u - v\|_1a^{L-\ell}}{\Delta} =
   \sqrt{d} \cdot \frac{\|u - v\|_1}{\sqrt{d}\Delta/a^{L-\ell}}
   \le d \cdot \frac{\|u - v\|_2}{\Delta_\ell},
   \]
   where the last inequality follows from the inequality $\|x\|_1 \leq
   \sqrt{d}\|x\|_2$, which holds for any $x$. 

 \item The degree of the cell containing all points is bounded by $(a+1)^d$.
   The degree of subcells is bounded by $a^d$. 
\end{enumerate}

To label the set $S$ by a hierarchical partition, sampled as above, we
need to first shift $S$ and compute $\Delta$ so that $S \subseteq [0,
\Delta]^d$ as above. For this, we just need to compute the smallest
and largest coordinate of any point in $S$, which can be done in
$O(\log_s n)$ rounds of MPC. Then, to label each $u \in S$, we let an
arbitrary processor sample $r$ and broadcast it to all other
processors; for each $\ell$, $u$ is labelled by the sequence $\left(\left\lfloor
   \frac{(u_i-r_i)a^{L-\ell}}{\Delta}\right\rfloor\right)_{i = 1}^d$,
which uniquely identifies $\mycell_\ell(u)$. 

Next we define one notion of a good ordering. To compare two cells
$\mycell_1, \mycell_2$, we find the lowest $\ell$ for which they
belong to the same cell $\mycell \in \detpart_\ell$. If $\mycell_1 \in
\detpart_\ell$, then $\mycell_1 < \mycell_2$, and vice
versa. Otherwise if $\mycell_1, \mycell_2 \not \in \detpart_\ell$, we
consider the two cells $\mycell', \mycell'' \in \detpart_{\ell-1}$
such that $\mycell_1 \subseteq \mycell'$ and $\mycell_2 \subseteq
\mycell''$. If the centroid of $\mycell'$ precedes the centroid of
$\mycell''$ lexigraphically, then $\mycell_1 < \mycell_2$, and vice
versa. This ordering can be seen as a pre-order on the natural tree
structure associated with $\detpart$, where the children of each node
are ordered lexicographically. It is straightforward to verify that
the hierarchical property is satisfied. Two cells, represented by
their centroids and level $\ell$ can be compared in $O(d\lev)$ time.
\end{proof}

\subsection{Bounded Ratio for MST and EMD}
\label{sec:aspectRatio}

We now show why we can assume that our dataset has a bounded ratio in
the MST and EMD applications. In particular, we show how to reduce an
arbitrary dataset to one where the aspect ratio is bounded by a
polynomial in $n$, while incurring a multiplicative error of $1+\eps$
only. To be precise, we note that we assume that original points are
from a set $[0,\Delta]^d$, where the machine word size is $O(\log
\Delta)$.

We employ standard reductions, see, e.g.,~\cite{A98-TSP, I07}. However, we
need to make sure that the reductions can be executed in the MPC
framework. Both run in $O(\log_s n)$ parallel time.

\paragraph{Euclidean MST problem.} 
The reduction first computes an approximation of the diameter of $S$
under the metric $\rho$. To compute the approximation, we pick an arbitrary
point $w \in S$ and compute $D = 2\max_{v \in S}{\rho(w,v)}$. Now
discretize all coordinates of all points to multiples of $\eps
D/(8\sqrt{d}n)$. Since, by the triangle inequality, $\max_{u, v\in
  S}\rho(u,v) \leq \max_v \rho(u, w) + \rho(w, v) \leq dD$, the aspect
ratio becomes at most $8\sqrt{d}n/\epsilon$. 

Consider the MST $T'$ computed on the modified input; we can construct a
tree $T$ for the original input by connecting points rounded to the
same point using an arbitrary tree, and then connecting these trees
as in $T'$. The cost of $T$ is at most the cost of $T'$ plus $\eps D/4$,
since the distance between any two points rounding to the same edge is
at most $\eps D/(4n)$. By an analogous argument, $\cost{T'} \le
\cost{\topt} + \eps D/4$, where $\topt$ is the MST for the original
input. Therefore, $\cost{T'} \leq \cost{\topt} + \eps D/2$. Notice that
$\cost{\topt} \geq \max_{u, v\in S}{\rho(u,v)} \geq D/2$, and it
follows that $\cost{T'} \leq (1+\eps) \cost{\topt}$.

All these operations are straight-forward to implement in the MPC model.

\paragraph{EMD and transportation problems.}
Our reduction works for the case when the maximal demand is
$U=n^{O(1)}$ (and is lower bounded by 1). We use the reduction from
\cite{I07}. For this, we first compute a value $M$ which is a
$A=O(\log n)$ approximation to the cost of EMD/transportation (suppose
$M$ is an overestimate). Note that we can accomplish this by running
the linear sketch of \cite{IT}, which we can implement in MPC model in
a straight-forward way. Next we round each coordinate of each point to
the nearest multiple of $\eps M/(AU\sqrt{d}n)$; this incurs an error
of at most $\eps M/A$ overall, which is at most an $\eps$ fraction of
the transportation cost.

Finally, we impose a randomly shifted grid, with side-length of
$10M$. Each cell defines an (independent) instance of the problem,
with aspect ratio $10M/(\eps M)\cdot
AU\sqrt{d}n=10AU\sqrt{d}n/\eps=n^{O(1)}$. \cite{I07} shows that, with probability at
least $0.9$, the cost of overall solution is equal to the sum of the
costs of each instance. 

It just remains to show how to solve all the instances in
parallel. First, we solve all instances of size less than $s$: just
distribute all the instances onto the machines, noting that each
instance can be just solved locally. Second, for instances of size
more than $s$, we assign them to machines in order. Formally, we just
assign to each point its cell number, and then sort all points by cell
number (using \cite{GSZ11-sorting}). This means that each instance is
assigned to a contiguous range of machines. Now solve each instance in
parallel. Note that each machine solves at most two instances at the
same time.






\section{Algorithms for Bounded Doubling Dimension}
\label{sec:doubling}

In this section we show an efficient algorithm for the minimum spanning tree problem for metrics with bounded doubling dimension. Before we prove the main result, we introduce useful tools such as an algorithm for sampling nets, and an algorithm for constructing hierarchical distance-preserving partitions.

\subsection{Constructing a $(\delta,\delta/4)$-net}

We now show how to sample a $(\delta,\delta/4)$-net. Let us state an auxiliary lemma about uniform sampling from a collection of sets in search for their representatives.


\begin{lemma}\label{lem:sampling_sets}
Let $S$ be a set of size $n$. Let sets $S_1,\ldots,S_k \subset S$ be a
partition of $S$, i.e., each element in $S$ belongs to exactly one
$S_i$.  Let $X$ be a subset of $S$ created by independently selecting
each element in $S$ with probability $p\ge \min\{\frac{k^2}{n} \cdot
\ln(2k/\gamma),1\}$, where $\gamma > 0$.  With probability $1-\gamma$, the
following events hold.
\begin{itemize}
\item The total size of sets $S_i$ that are not hit is bounded:
$$\sum_{i:S_i \cap X = \emptyset} |S_i| \le n / k.$$
\item The size of $X$ is bounded: $|X| \le \min\{2pn,n\}$
\end{itemize}
\end{lemma}

\begin{proof}
If $p = 1$ or $k=1$, the lemma is trivial. Assume therefore that $p = \frac{k^2}{n} \ln \frac{2k}{\gamma} <
1$ and $k \ge 2$. For each $S_i$, the probability that no element of
$S_i$ is selected is bounded by $(1-p)^{|S_i|} \le
e^{-k^2|S_i|\ln (2k / \gamma)/n}$. If $|S_i| \ge n/k^2$, then the
probability that $S_i \cap X = \emptyset$ is bounded by
$e^{-\ln(2k/\gamma)} = \gamma/(2k)$. By the union bound, with probability
$1-\gamma/2$, the only sets $S_i$ that do not intersect with $X$ have
size bounded by $n/k^2$, and therefore their union is bounded by
$k\cdot n/k^2 = n/k$.

By the Chernoff bound, the probability that $|X| > 2pn$ is bounded by $e^{-pn/3} \le e^{-k^2\cdot \ln(2k/\gamma) / 3} \le e^{-\ln(2k/\gamma)} = \gamma/(2k)$. Thus by the union bound both desired events hold with probability at least $1-\gamma$.
\end{proof}

For time and space efficiency, our algorithms for bounded doubling
dimension make extensive use of a data structure for the dynamic
approximate nearest neighbor search (ANNS) problem due to Cole and
Gottlieb~\cite{CG-NN}. We state its guarantees next.

\begin{theorem}[\cite{CG-NN}]\label{thm:dd-anns}
  Let $M = (S, \ell_2^d)$ be a metric space with doubling dimension
  $d$, and let $|S| = n$. There exists a data structure for the
  dynamic $\epsilon$-ANNS problem (see Definition~\ref{def:dyn-anns})
  with $O(n)$ space, and $2^{O(d)}\log n$ update time, and
  $2^{O(d)}\log n + \eps^{-O(d)}$ query time.
\end{theorem}

\begin{lemma}\label{lemma:find_net}
Let $M = (S,\rho)$ be a metric space with doubling dimension $d$,
where $|S| = n > 1$, which can be covered with a ball of
radius $R$. The set $S$ is given as input and the distance
$\rho(\cdot,\cdot)$ between each pair of points can be computed in
constant time. Let $m$ be the number of machines and $\gamma\in(0,1)$ is a parameter such that $n/m \ge C\cdot\log^3 (n/\gamma)$, where $C$ is a sufficiently large constant.
Let $t =
\left\lfloor \frac{\log (n/m)}{3d} \right\rfloor$ and $\delta = 8R/2^{t}$.
There is a parallel algorithm that computes a $(\delta,\delta/4)$-net of size at
most $\sqrt[3]{n/m}$ using $m$ machines. With probability $1-\gamma$, the following properties hold:
\begin{itemize}
\item each machine uses at most $O(n/m)$ space,
\item the number of rounds is $O(\log_{(n/m)}
  n \cdot (1+\log_{(n/m)}m))$,
\item the computation time of a machine per round is bounded by $2^{O(d)}\cdot(n/m)\cdot \log (n/m)$.
\end{itemize}
\end{lemma}

\begin{algorithm}[t]
\caption{An idealized algorithm for constructing a $(\delta,\delta/2)$-net for a metric $M=(S,\rho)$.}\label{alg:net_ideal}
$S':=S$\\
$U := \emptyset$\\
\While{$S' \ne \emptyset$}{
$U':= \mbox{sample each point in $S'$ independently with probability $\min\{1,n/(m\cdot|S'|)\}$}$\\
$U'' := \mbox{maximal subset of $U'$ such that no two points
  are at distance less than $\delta/2$}$.%
\nllabel{step_ideal:maximal}\\
Remove from $S'$ all points that are at distance at most $\delta$
from a point in $U''$.%
\nllabel{step_ideal:remove}
\\
$U := U \cup U''$\\
}
\Return{U}
\end{algorithm}

\begin{algorithm}[t]
\caption{An efficient algorithm for constructing a $(\delta,\delta/4)$-net for a metric $M=(S,\rho)$.}\label{alg:net_efficient}
$S':=S$\\
$U := \emptyset$\\
\While{$S' \ne \emptyset$\nllabel{step_efficient:main_loop}}{
$U':= \mbox{sample each point in $S'$ independently with probability $\min\{1,n/(m\cdot|S'|)\}$}$\\
$U'':=\emptyset$\\
Create a 2-ANNS data structure $\mathcal D_1$ containing points in $U'$.%
\nllabel{step_efficient:maximal_NN}\\
\While{$U' \ne \emptyset$\nllabel{step_efficient:maximal}}{
$u:=\mbox{any point in $U'$}$\\
As long the approximate nearest neighbor $v$ for $u$ returned by $\mathcal D_1$ is at distance at most $\delta/2$ from $u$, remove $v$ from both $\mathcal D_1$ and $U'$\\
$U'' := U'' \cup \{u\}$}
Create a $4/3$-ANNS data structure $\mathcal D_2$ containing points in $U''$.%
\nllabel{step_efficient:remove_NN}\\
\ForEach{$v\in S'$\nllabel{step_efficient:remove}}{If the approximate nearest neighbor for $v$ given by $\mathcal D_2$ is at distance at most $\delta$ from $v$, remove~$v$ from $S'$.}
$U := U \cup U''$\\
}
\Return{U}
\end{algorithm}

\begin{proof}Our starting point is 
Algorithm~\ref{alg:net_ideal}. In this algorithm, a point is removed from
the set $S'$, which is initially a copy of $S$, if and only if it is
at distance at most $\delta$ from a point in $U''$. Since the final
$U$ is a union of all such $U''$, every point in $S$ is at distance at
most $\delta$ from a point in the final $U$. This implies that the
algorithm returns a $\delta$-covering. Additionally, no two points in $U$ can be at distance less than $\delta/2$, so the final set is a $\delta/2$ packing as well.

A parallel implementation of Algorithm~\ref{alg:net_ideal} can be used
to obtain all desired properties with the exception of the bound on
the running time. In order to obtain an efficient algorithm for
creating nets, we slightly weaken the quality of the net that we may
obtain. In Steps~\ref{step_ideal:maximal} and~\ref{step_ideal:remove}
of Algorithm~\ref{alg:net_ideal}, we remove points up to a specific
distance. A direct implementation of these steps, by comparing all
pairs of distances would result in an algorithm with total work at
least $\sigma|S|$, where $\sigma$ is the size of the final $U$. If
$|S| = n$, and the number of machines $m$ is at most
$n^{1-\Omega(1)}$, the total work could be as large as
$n^{1+\Omega(1)}$, i.e., significantly superlinear. We therefore
replace these two steps with weaker versions. 
In particular, we use approximate nearest neighbor search to find points at a particular distance; we show that the relaxation resulting from the approximate nearest neighbor is still sufficient in our case. 
The modified algorithm is presented as
Algorithm~\ref{alg:net_efficient} and uses the data structure from
Theorem~\ref{thm:dd-anns} to implement the ANNS data structure.

In Algorithm~\ref{alg:net_efficient}, each pair of different points in $U''$ is at distance more than $\delta/2$. Moreover, once a point is included in $U''$, and therefore also in $U$, all points at distance at most $\frac{3}{4}\delta$ are removed from $S'$. Hence, the final set $U$ is a $\delta/4$-packing. It is also a $\delta$-covering, because no point in $S$ is removed from $S'$, unless there is a point in $U''$ (and hence in $U$) at distance at most $\delta$. Summarizing, the algorithm produces a $(\delta,\delta/4)$-net.

Furthermore, observe that by definition, $M$ can be covered with at
most $k$ balls $B_1$, \ldots, $B_{k}$ of radius $\delta/8=R/2^t$, where $k
\le (2^d)^t \le \sqrt[3]{(n/m)}$. Because in the final $U$, two different points are at distance more than $\delta/4$, only at most one point belongs to each ball $B_i$.
Hence, the size of the net is bounded by $k \le \sqrt[3]{(n/m)}$.

One of the main challenges is to show that the algorithm performs a small
number of iterations. Let $C$ be a constant such that $C \cdot \log^3
n \ge (\log(2n^2/\gamma))^3$. Every iteration of the loop in Line~3 is
likely to decrease the number of points in $S'$ by a factor of at
least $\sqrt[3]{n/m}$. Consider the points remaining in $S'$ at the
beginning of the loop evaluation. They can be partitioned into at most
$k$ sets of diameter at most $\delta/4$ each. The sets correspond to
the covering by balls $B_i$, where each point is arbitrarily assigned
to one of the balls covering it. If a given set has a point $v$
selected for $U'$ in the loop in Step~\ref{step_efficient:maximal}, then the set is completely removed
in this iteration of the loop from $S'$. This is obvious to see if $v$ or any
other point in the set is selected for $U''$. Otherwise, there must be
a point $u \in U''$ at distance at most $\delta/2$ from $v$. Because
the diameter of the set is bounded by $\delta/4$, all points in the
set are at distance at most $\frac{3}{4} \delta$ from $u$ and will be removed from
$S'$ in the loop in Step~\ref{step_efficient:remove}.  We apply Lemma~\ref{lem:sampling_sets} on $S'$
with $k=\sqrt[3]{n/m}$ and $p=k^3/|S'|$. With probability at least $1-\gamma/n$, in a given iteration of the loop in Step~\ref{step_efficient:main_loop}, both 
the size of $U'$ is
bounded by $2n/m$ and the size of $S'$ decreases by a factor of at
least $\sqrt[3]{n/m}$. Let $\mathcal E$ be the event that this happens in all iterations of the loop. By the union bound, $\mathcal E$ occurs with probability at least 
$1-\gamma$. If $\mathcal E$ occurs, the algorithm finishes in $\log_{\sqrt[3]{n/m}} n
= 3\log_{(n/m)} n$ iterations of the main loop.

It remains to describe a parallel implementation of Algorithm~\ref{alg:net_efficient}. We assume that the input to the algorithm is the set $S'$
of points which is distributed across all $m$ machines with each having no more than $O(n/m)$ of them. In order to generate the set $U'$ collectively, each machine selects each point independently at random with the desired probability, which
takes $O(n/m)$ time. All points sampled in the process are sent to a single machine. If the event $\mathcal E$ occurs, the total number of them is $O(n/m)$, i.e., they can be stored on a single machine. The machine constructs $U''$ from $U'$ as in Steps~\ref{step_efficient:maximal_NN} and~\ref{step_efficient:maximal} of Algorithm~\ref{alg:net_efficient}. This takes at most $2^{O(d)}\cdot|U'|\cdot \log|U'|$ time, where the running time is dominated by the nearest neighbor data structure. If the event $\mathcal E$ occurs, this takes at most $2^{O(d)}\cdot (n/m)\cdot \log(n/m)$ time.
Once $U''$ is computed, it is broadcasted to all machines. Recall that $|U''| \le \sqrt[3]{n/m}$. This implies that all machines can receive $U''$ in at most $1+\log_{(n/m)/\sqrt[3]{n/m}} m = O(1+\log_{(n/m)}m)$ rounds of communication with no machine sending more than $O(n/m)$ information per round and with at most linear running time per round. After that each machine first constructs a nearest neighbor data structure as in Step~\ref{step_efficient:remove_NN} of Algorithm~\ref{alg:net_efficient} and runs the {\bf foreach} loop in Step~\ref{step_efficient:remove} on its share of points remaining in $S'$. This takes at most $2^{O(d)} \cdot (n/m)\cdot \log(n/m)$ time.
\end{proof}

\subsection{Constructing a distance-preserving partition}

\begin{algorithm}[t]
\caption{Partitioning algorithm ($S'\subseteq S$ is a $(\delta,\delta/4)$-net)}\label{alg:partition}
Let $v_1,v_2,\ldots, v_{|S'|}$ be a random permutation of points in $S'$\\
Let $r$ be selected uniformly at random from $[\delta,2\delta]$\\
$T:=S$\\
\For{$i=1,\ldots,|S'|$}{
$C_i := T \cap B(v_i,r)$\\
$T := T \setminus C_i$\\
}
\Return{$C_1,C_2,\ldots, C_{|S'|}$}
\end{algorithm}

In order to partition cells into subcells, we use
Algorithm~\ref{alg:partition}. We now prove that it provides a
partition with the desired properties. The proof uses an argument borrowed from \cite{Talwar04} and earlier works (see the references in \cite{Talwar04}).  

\begin{lemma}\label{lem:partition}
Let $(S,\rho)$ be a metric space with bounded doubling dimension $d$. Let $S' \subset S$ be a $(\delta,\delta/4)$-net. Algorithm~\ref{alg:partition} produces a partition of $S$ such that
the diameter of each cluster in the partition is bounded by $4\delta$ and for any two points $x$ and $y$, the probability that they are split is bounded by $O(d)\cdot \rho(x,y)/\delta$.
\end{lemma}

\begin{proof}
Observe that each point $u \in S$ is assigned to the first $C_i$ such that $\rho(u,v_i) \le r$. Moreover, each $x$ is assigned to at least one $C_i$, because $S'$ is a $\delta$-covering and $r \ge \delta$. Hence $C_1$, \ldots, $C_{|S'|}$ is a partition.

To prove that the diameter of each $C_i$ is bounded by $4\delta$, observe first that the distance of each point in $C_i$ to the corresponding $v_i$ is bounded by $r\le 2\delta$. Therefore, it follows from the triangle inequality that the distance between any pair of points in $C_i$ is at most $4\delta$. 

It remains to prove that close points $u$ and $w$ are separated in the
partition with bounded probability. Without loss of generality, we
assume that $\rho(u,w) \le \delta$, because otherwise the bound on
probability is trivial.  We say that $q \in S'$ {\em separates} $u$
and $w$ if when $w$ is considered in the sequence $v_1$, \ldots,
$v_{|S'|}$, both $u$ and $w$ are still in $T$ and exactly one of them
is at distance at most $r$ from $q$.  Clearly, $u$ and $w$ are
assigned to different clusters if and only if a point in $S'$
separates them.  Let $S'' = S' \cap B(u,3\delta)$. Note that, since
$S'$ is a $\delta$-covering, if $q\not\in S''$, then it cannot
separate $u$ and $w$, because both $u$ and $w$ are at distance greater
than $r\le 2\delta$ from such a $q$. We now bound the size of
$S''$. By the definition of doubling dimension, $B(u,3\delta)$ can be
covered with at most $(2^d)^5$ balls of radius $(3/32) \delta <
\delta/8$. Let $X\subseteq S$ be the centers of those balls. For each
point $q$ in $S''$, at least one point in $X$ is at distance less than
$\delta/8$ from $q$. Furthermore, no point in $X$ can be at distance
less than $\delta/8$ from two distinct points in $S''$ simultaneously,
because $S''$ is a subset of a $\delta/4$-packing, and therefore each
pair of points in $S''$ is at distance at least $\delta/4$ from each
other. Therefore, there exists an injective mapping that maps each
point $q$ in $S''$ to the unique point in $X$ that at distance at most
$\delta/8$ from $q$; the existence of such a mapping shows that $|S''|
\le |X| \le 32^d$.

Let $q_1$, \ldots, $q_{|S''|}$ be the points in $S''$ in
non-decreasing order of $\min\{\rho(q_i,u),\rho(q_i,w)\}$. We now
bound the probability that each of them separates $u$ and $w$. First
the probability that the ball $B(q_i,r)$ contains exactly one of them
is bounded as follows: 
\begin{align*}
\Pr[|B(q_i,r) \cap \{u, w\}| = 1] &= \Pr[r \in
(\min\{\rho(q_i,u),\rho(q_i,w)\}, \max\{\rho(q_i,u),\rho(q_i,w)\})] \\
&\leq Pr[(\min\{\rho(q_i,u),\rho(q_i,w)\}, \rho(u,w) +
\min\{\rho(q_i,u),\rho(q_i,w)\})]  = \frac{1}{\delta}\rho(u,w),
\end{align*}
where the second inequality follows because, by the triangle
inequality, 
\[
\max\{\rho(q_i,u),\rho(q_i,w)\}\leq \rho(u,w) +
\min\{\rho(q_i,u),\rho(q_i,w)\}.
\]
Conditioned on this event, the probability that $q_i$ separates $u$
and $w$ is at most $1/i$, because if any $q_j$ with $j < i$ appears
before $q_i$ in the random permutation $v_1$, \ldots, $v_{|S'|}$, then
it also removes at least one of the points from $T$, so $q_i$ cannot
separate them. Therefore, $u$ and $w$ are assigned to different
clusters with probability at most
$$\sum_{i=1}^{|S''|}\frac{\rho(u,w)}{\delta}\cdot\frac{1}{i} = 
\frac{\rho(u,w)}{\delta}\cdot H_{|S''|}
= \frac{\rho(u,w)}{\delta}\cdot O(\log |S''|)
= \frac{\rho(u,w)}{\delta}\cdot O(d),$$
where $H_k$ is the $k$-th harmonic number. This finishes the proof.
\end{proof}

The proof of the following lemma describes how to apply Lemma~\ref{lemma:find_net} and Lemma~\ref{lem:partition} in order to show that there is an efficient algorithm for computing hierarchical partitions.

\begin{lemma}\label{lem:doubling_hierarchy}
  Let $M = (S,\rho)$ be a metric space with bounded doubling dimension
  $d$, where $|S| = n > 1$, such that the entire space can be covered
  with one ball of radius $R$.  Assume the set $S$ is given as input, and the
  distance $\rho(\cdot,\cdot)$ between each pair of points can be
  computed in constant time. Let $m$ be the number of machines, $\gamma
  \in (0,1)$, and assume that $n/m \ge \max\{K\cdot \log^3
  (n/\gamma),2^{36\cdot d}\}$, where $K$ is a sufficiently large
  constant. Let $t = \left\lfloor
    \frac{\log(n/m)}{6d}\right\rfloor$. There is an algorithm that
  computes a $(2^{5-t},O(d),\sqrt[3]{n/m})$-distance-preserving
  hierarchical partition $P=(P_0,\ldots,P_L)$ with $L$ levels and
  approximation $2$. With probability $1-\gamma n L$, the following
  properties hold:
\begin{itemize}
\item no machine uses more than $O(L \cdot n/m)$ space,
\item the total number of rounds is bounded by $O(L \cdot \log_{(n/m)} n \cdot (1+\log_{(n/m)} m))$,
\item the computation time of each machine per round is bounded by $2^{O(d)} \cdot (n/m)\cdot\log(n/m)$.
\end{itemize}
\end{lemma}

\begin{proof}
  We first sketch our algorithm without discussing the details of its
  parallel implementation. The algorithm computes consecutive
  partitions $P_i$ in $L$ phases. Initially, $P_L = \{S\}$, where the
  diameter of $S$ is bounded by $2R$ by assumption. In the $i$-th
  phase, we compute $P_{L-i}$. Let $a = 2^{5-t}$ and let $\Delta_\ell
  = 2a^{L-\ell}R$. By induction, we show that each partition $P_\ell$
  is such that all cells in $P_\ell$ have diameter bounded by
  $\Delta_\ell$. In the $i$-th phase, we partition cells in
  $P_{L-i+1}$ into cells that constitute a partition
  $P_{L-i}$. Consider a cell $\mycell \in P_{L-i+1}$. $\mycell$ can
  be covered by a ball of radius $\Delta_{L-i+1}$. Moreover, by
  Lemma~\ref{lem:ddim_restricted}, the doubling dimension of the
  metric space restricted to $\mycell$ is bounded by $2d$. We start by
  generating a $(\Delta_{L-i+1} \cdot 2^{3-t},\Delta_{L-i+1} \cdot
  2^{1-t})$-net, i.e., a $(\Delta_{L-i}/4,\Delta_{L-i}/16)$-net, for
  each such cell $\mycell$ in $P_{L-i+1}$ in parallel. To achieve
  this goal, we use the algorithm of Lemma~\ref{lemma:find_net}. Then
  we partition each cell, using the corresponding net. For each cell
  $\mycell$, we generate a random permutation of the net and a random
  parameter $r \in [\Delta_{L-i}/4,\Delta_{L-i}/2]$ as in
  Algorithm~\ref{alg:partition}.  Each cell is then partitioned into
  subcells as in Algorithm~\ref{alg:partition} for this choice of
  $r$. Note that each subcell is a subset of a ball of radius
  $\Delta_{L-i}/2$. Hence, each subcell has diameter bounded by
  $\Delta_{L-i}$ as desired. Moreover,  the partition has size
  bounded by the size of the net, which, by Lemma~\ref{lemma:find_net}
  has size at most $\sqrt[3]{n/m}$. 

  Observe that $a\ge 2$, because $n/m \ge 2^{36\cdot d}$. Therefore,
  the probability that for $\Delta_\ell \ge \rho(u,v)$, two points
  $u,v \in S$ are separated is bounded by $\sum_{i=\ell}^{L}
  \frac{O(d)\cdot\rho(u,v)}{\Delta_i\cdot a^{i-\ell}} = O(d) \cdot
  \rho(u,v)/\Delta_i$ due to the properties of the partition generated
  by Algorithm~\ref{alg:partition}, which were shown in
  Lemma~\ref{lem:partition}.  This shows that the algorithm produces a
  $(2^{5-t},O(d),\sqrt[3]{n/m})$-distance-preserving hierarchical partition.

In each phase we first construct a net for each cell. For small cells
that fit in the memory of a single machine, this can be done
locally.\footnote{If points in a small cell happen to be shared by two
  machines, we can send them all to one of them.} For larger cells, we
run the algorithm described in Lemma~\ref{lemma:find_net}, where we
can still guarantee that the ratio $n/m$ is sufficiently large, which
in the worst case may require limited reshuffling of points between
machines. Once we are done, for each cell, the machine holding the
net, generates its random permutation and a random parameter $r$ as in
Algorithm~\ref{alg:partition}. These objects are next broadcasted to
all machines that hold points from the cell. Since the size of the net
is bounded by $\sqrt[3]{n/m}$, this requires at most $O(1+\log_{n/m}
m)$ rounds of communication with linear work and communication per
machine per round. After that each machine computes the partition of
points it holds as in Algorithm~\ref{alg:partition}. An efficient
implementation uses the $\epsilon$-ANNS data structure from
Theorem~\ref{thm:dd-anns} with approximation factor of $1+\epsilon =
2$. 
We first insert into the data structure all
points in the net. Then for each point $w$ in the cell, we want to
find all points at distance at most $r$. This can be done by finding
approximate nearest neighbors of $w$ in the data structure and
removing them as long as the distance of the neighbor is bounded by
$2r$. All points at distance at most $r$ from $w$ have to be on the
list of the neighbors we have found. At the end, we insert the
neighbors back into the data structure. To bound the number of
approximate nearest neighbors we consider, recall that the net is a
$\delta/4$-packing and the doubling dimension of the cell is bounded
by $2d$. Therefore, the ball of radius $2r$ contains at most
$2^{O(d)}$ net points via the standard argument. The computation time
for this step may be superlinear, but is at most $2^{O(d)}\cdot
(n/m)\cdot\log (n/m)$ per machine. \junk{After each point is assigned
  to a subcell, in the next $O(1+\log_{n/m} m)$ rounds we compute the
  number of points in each subcell and broadcast an assignment of
  exclusive subranges of $[n]$ to subcells and points in
  machines. This allows for sending all points to new machines in one
  round.} After this, we have computed an assignment of subcells to
points and we can proceed to the next lower level of $\detpart$.

The running time of the algorithm per machine per round is dominated by the rounds that require using the nearest neighbor data structures, which is at most $2^{O(d)}(n/m) \log (n/m)$. The communication complexity is dominated by reshuffling of points at the end of each phase, which requires sending 
$O(L\cdot n/m)$ information. Each point is sent with a list of cells it belongs to at each level. The bounds on space and hold with probability at least $1-\gamma n L$ via the union bound.
\end{proof}

\subsection{MST for Bounded Doubling Dimension}

\begin{theorem}
Let $M = (S,\rho)$ be a metric space with bounded doubling dimension $d\ge 1$, where $|S| = n > 1$. Let $\eps \in (0,1/3)$.
Assume the set $S$ is given as input, and the distance
$\rho(\cdot,\cdot)$ between each pair of points can be computed in
constant time. Let $m$ be the number of machines, and assume that $$n/m \ge \max\left\{K\cdot \log^3 n,2^{36\cdot d},\left(\left(1+\frac{\log(n/\eps)}{\log(n/m)}\right)\cdot \frac{2^d}{\eps}
\right)^{C \cdot d}\right\},$$ where $K$ is a sufficiently large constant. Let $t = \left\lfloor \frac{\log(n/m)}{6d}\right\rfloor$. There is an algorithm that computes a spanning tree of the set of the points such that:
\begin{itemize}
\item the expected cost of the tree is $1+\eps$ times the cost of the minimum spanning tree,
\item with probability $1-1/n^3$,
\begin{itemize}
\item the total number of rounds is bounded by $O\left(\log_{n/m} n \cdot L \cdot  \left(1+\log_{(n/m)}m\right)\right)$, where $L = 1+O\left(
d\cdot\frac{\log(n/\eps)}{\log(n/m)}\right)$,
\item the computation time per machine per round is bounded by $O(2^{O(d)} \cdot (n/m)\cdot\log^{O(1)}(n/m)
$,
\item the total space per machine is bounded by $O(L \cdot (n/m))$,
\item the total communication per machine per round is bounded by $O(L \cdot (n/m))$.
\end{itemize}
\end{itemize}
\end{theorem}

\begin{proof}
We first use the algorithm of Indyk~\cite{Indyk99} to compute a $2$-approximation for the diameter of the set of points. We pick a single point $u$, and send it to all machines. Each machine computes in linear time the maximum distance of its points to $u$. The maximum of these values is sent back to a single machine. This takes at most $1+O(\log_{n/m}m)$ rounds with at most linear computation and communication per round per machine. Let $R$ be the computed value. Clearly, the diameter of the set of points lies in $[R,2R]$.

Let $\rho(\topt)$ be the cost of the minimum spanning tree. Consider
now the following procedure. Pick any $\eps R/(2n)$-covering $S'$ from
$S$ and then remove $S' \setminus S$ from the input.  Observe that
this changes the cost of any tree by at most $(n-1) \cdot \eps R/(2n)
\le \eps R/2$, because each removed point can be connected to an
arbitrary point at distance at most $\eps R/(2n)$. This implies that
if we compute the minimum spanning tree in the modified point set, the
corresponding tree for the original point set cannot cost more than
$\rho(\topt) + \eps R$. Furthermore, since $\rho(\topt)\ge R$, it
suffices to compute a minimum spanning tree for the modified point set
to obtain a $(1+\eps)$-approximation for the original problem.

We now wish to apply Lemma~\ref{lem:doubling_hierarchy} in order to
compute a hierarchical partition in which the diameter of cells at the
lowest level is bounded by $\eps R/(6n)$, at which point we can
connects points arbitrarily as observed above. The number of
levels of the partition we need can be bounded by
$\left\lceil\frac{\log(6n/\eps)}{t}\right\rceil$, where $t =
\left\lfloor\frac{\log(n/m)}{6d}\right\rfloor$. Due to our
requirements on the $n/m$ ratio, one can show that $L$, the number of
levels, is of order at most
$1+O\left(d\cdot\frac{\log(n/\eps)}{\log(n/m)}\right)$. We run the
algorithm of Lemma~\ref{lem:doubling_hierarchy} to compute the desired
$(2^{5-t},O(d),2)$-distance-preserving hierarchical partition with $L$
levels, with that algorithm's failure probability parameter $\gamma$ set to
$1/n^{5}$. The algorithm obtains the desired runtime bounds with
probability $1-n^3$. In particular, this means that running the
algorithm requires at most $O\left(\log_{n/m} n \cdot L \cdot
  \left(1+\log_{(n/m)}m\right)\right)$ rounds of communication.

By Theorem~\ref{thm:mst}, the expected
cost of the tree output by the Solve-and-Sketch algorithm with unit
step Algorithm~\ref{alg:unit-step} is bounded by $(1+\eps' \cdot O(L)
\cdot 2^{O(d)}) \cost{\topt}$ for some parameter $\eps'$. By setting $\eps'=
\eps/(L\cdot 2^{K_1 \cdot d})$ for some large constant $K_1$, the
expected cost of the tree becomes at most $(1+\eps)\cost{\topt}$. 

Next we want to discuss how to execute Algorithm~\ref{alg:unit-step}
in the cells of the resulting partition. For the sake of time and
space efficiency, we use the dynamic $\epsilon$-ANNS data structure
from Theorem~\ref{thm:dd-anns}. Algorithm~\ref{alg:unit-step} is
implemented exactly as in the proof of
Lemma~\ref{lm:unit-step-compl}, but substituting the data structure
fro Theorem~\ref{thm:dd-anns} for the Euclidean space $\epsilon$-ANNS
data structure. I.e.~we use Eppstein's reduction in
Theorem~\ref{thm:ccp} to construct an data structure for the dynamic
$\epsilon$-CCP problem (Definition~\ref{def:ccp}). Then we simulate
Algorithm~\ref{alg:unit-step} using the $\epsilon$-CCP data structure,
and each time we need to merge two connected components, we recolor
the points of the smaller connected component to the color of the
larger one. Each point is recolored at most $O(\log (n/m))$ times, and
the overall time complexity is $2^{O(d)}(n/m)\log^{O(1)}(n/m)$ by
Theorems~\ref{thm:ccp} and \ref{thm:dd-anns}. The space complexity is
$O((n/m)\log (n/m))$ by 

Each execution of the unit step (Algorithm~\ref{alg:unit-step}) needs
to return an $\eps'^2\Delta_\ell$-covering (recall that $a=2^{5-t}$,
$\Delta_\ell = 2Ra^{L-\ell}$). We prove the bound on output size by
induction. The coverings for cells in $P_0$ are simply single points,
so the desired bound holds. Consider $\ell > 0$. After connecting
components has been finished, we combine the
$\eps'^2\Delta_{\ell-1}$-coverings for subcells together. Let $U$ be
their union. $U$ is a $\eps'^2\Delta_{\ell-1}$-covering for the
current cell.
We sparsify the set of points, creating a covering $V'$, as follows. Initially, $V' = \emptyset$. We pick an arbitrary point $v \in U$. 
We add $v$ to $V'$ and remove all points at distance at most $\eps'^2\Delta_{\ell}/2$
from $U$. 
Once we are done, if $U$ is non-empty, we repeat the procedure for a new $v$. The resulting set $V'$ is a $\eps'^2\Delta_{\ell}$-covering of the current cell, because each point in the cell is at distance at most $\eps'^2\Delta_{\ell-1}+\eps'^2\Delta_{\ell}/2 \le \eps'^2\Delta_{\ell-1} (a + 1/2) \le \eps'^2\Delta_{\ell-1}$ from some point in $V'$,
which finishes the induction.

Observe that due to how $V'$ was constructed, the minimum distance
between points is greater than $\eps'^2\Delta_{\ell}/2$. We can use
this property to bound the size of each $V'$ constructed in the
process. $V'$ can be covered with a ball of radius $\Delta_\ell$ in
the original metric space, which implies it can be covered with at
most $2^{d\cdot\lceil\log (4/\eps'^2)\rceil}$ balls of radius
$\Delta_\ell/4$. Each of these balls can contain at most one point in
$V'$, which means that the size of the covering is bounded by
\[(1/\eps')^{O(d)} = (L\cdot 2^d/\eps)^{O(d)} =
\left(\left(1+O\left(d\cdot\frac{\log(n/\eps)}{\log(n/m)}\right)\right)\cdot
  2^d/\eps \right)^{O(d)}=
\left(\left(1+\frac{\log(n/\eps)}{\log(n/m)}\right)\cdot 2^d/\eps
\right)^{O(d)}.\] 
In order to satisfy the output size requirements of
Theorem~\ref{thm:unitstep}, it suffices
that \[(\sqrt[3]{n/m})\log(n/m) \ge
\left(\left(1+\frac{\log(n/\eps)}{\log(n/m)}\right)\cdot 2^d/\eps
\right)^{K_2 \cdot d}\] for some large constant $K_2$.

The existence of an efficient MPC implementation with the promised
time, space, communication, and round complexity then follows from
Lemma~\ref{lem:doubling_hierarchy} and Theorem~\ref{thm:unitstep}. 
\end{proof}

\section{Lower Bounds}
\label{sec:lb}
\subsection{Conditional Lower Bound for MST}

Computing connectivity for sparse graphs appears to be a hard problem
for the MPC/MapReduce model with a constant number of rounds
\cite{BKS13-communication}. Assuming the hardness of this problem
(i.e., that it cannot be solved in a constant number rounds), we show
that computing the exact cost of the minimum spanning tree in a
$O(\log n)$-dimensional space requires a super-constant number of
rounds as well.

\begin{theorem}
If we can exactly compute the cost of the minimum spanning tree
in $\ell_\infty^d$ for $d=100\log n$ in a constant number of rounds in the MPC
model, then we can also decide whether a general graph $G$ with $O(n)$
edges is connected or not in a constant number of rounds.
\end{theorem}

\begin{proof}
Assuming we have a constant-round algorithm for MST, we show how to
solve the connectivity problem. The input is a graph $G$ on $n$
vertices with $E=O(n)$ edges. Without loss of generality, no vertex in the graph is isolated.\footnote{To achieve this property, we can always add a shadow vertex $i'$ for each vertex $i$, and connect them.} This way the connectivity of edges and the connectivity of vertices become the same.

For each $i\in[n]$, pick a random vector $v_i\in\{-1,+1\}^d$. For each
edge $e=(i,j)$, we generate a point $p_e=v_i+v_j$. Then we solve MST
for the set of points $p_e$, where $e$ ranges over all edges $e$.

We claim that, with high probability, if the graph is connected, then MST cost is
$2(E-1)$. Otherwise, MST cost is $\ge 2(E-1)+2$.

\begin{claim}
With high probability, for each distinct $i,j, k,l$, we have:
\begin{itemize}
\item
$\|v_i-v_j\|_\infty=2$
\item
$\|(v_i+v_j)-(v_k+v_l)\|_\infty=4$.
\end{itemize}
\end{claim}
\begin{proof}
The claims follow immediately from an application of the Chernoff
bound. In the second case, note that each coordinate of
$(v_i+v_j)-(v_k+v_l)$ has a constant probability of being equal to
$\pm 4$.
\end{proof}

Now, if two edges $e,e'$ are incident, then we have that
$\|p_e-p_{e'}\|_\infty =2$, and otherwise, we have that
$\|p_e-p_{e'}\|_\infty=4$. Hence, if the graph $G$ is connected, then
MST cost is $2(E-1)$. Otherwise, the MST of all points $p_e$ must
connect two non-incident edges, increasing the cost to at least $2(E-1)+2$.
\end{proof}

\subsection{Query Lower Bound for Exact MST in Constant Doubling Dimension}

The lower bounds we present here and in Section~\ref{sec:lb_approx_general}
concern the complexity of the black-box distance query model. More specifically,
we assume that there is a distance oracle that each machine can query
in parallel in $O(1)$ time. In order to make a distance query, the machine
needs to have two identifiers, but it can hold only a limited number
of them at any time.

We use this model in our algorithm for MST for spaces with bounded
doubling dimension in Section~\ref{sec:doubling}. The goal of the
current section is to show that in order to obtain efficient
algorithms in this model for bounded doubling dimension one needs to allow for approximate solutions. 

Let $n$ be the total number of points. Let $s \ge 2$ be the number of
identifiers a machine can store and let $m$ be the number of
machines. Clearly, the total number of distances that can be queried
in a single round is $m\cdot \binom{s}{2} \le m\cdot s^2$. By showing
that the total number of queries has to be large, we show that
multiple rounds of computation are necessary.

We start with an auxiliary lemma.

\begin{lemma}\label{lem:find_in_array}
Consider a function $f:[n] \to \{0,1\}$, where $n \ge 5$, such that the value of $f$ equals 1 for exactly one $x \in [n]$. Correctly guessing $x$ such that $f(x)=1$ with probability greater than $1/2$ requires more than $n/4$ queries to $f$.
\end{lemma}

\begin{proof}
We write $x_\star$ to denote the argument for which $f(x_\star) = 1$. We consider the uniform distribution on possible inputs, i.e., $x_\star$ is selected uniformly at random from $[n]$. Consider any algorithm that makes at most $n/4$ queries. We show that it succeeds with probability at most $1/2$. The probability that the algorithm queries $f(x_\star)$ is at most $(n/4)/n = 1/4$. If it does not query $f(x_\star)$, there are at least $n-\lfloor n/4 \rfloor \ge 4$ unqueried arguments for which the function may equal 1. The probability that the algorithm correctly guess it is therefore bounded by $1/4$. The probability that the algorithm outputs the correct $x_\star$ is at most
$\Pr[\mbox{$f(x_\star)$ is queried}] + \Pr[\mbox{$f(x_\star)$ is not queried}] \cdot
\Pr[\mbox{A}|\mbox{$f(x_\star)$ is not queried}]
\le 1/4 + 1 \cdot 1/4 = 1/2.$\end{proof}

\begin{theorem}
Any algorithm computing the minimum spanning tree for a set of more than $6$ points of doubling dimension ${\le}\log 3$ with probability greater than $1/2$ requires $n^2/16$ distance queries and at least $n^2/(16ms^2)$ rounds.
\end{theorem}

\begin{proof}
Consider two sets of points: $A=\{a_1,\ldots,a_{n/2}\}$ and $B=\{b_1,\ldots,b_{n/2}\}$, each of size $n/2$. We now define a distance function $\delta$. For all $i,j \in [n/2]$, we have $\delta(a_i,a_j) = \delta(b_i,b_j) = |i-j|$. Furthermore, there is a pair of indices $i_\star,j_\star \in [n/2]$ such that $\delta(a_{i_\star},b_{j_\star}) = n$, and for all $i,j\in [n/2]$ such that $i \ne i_\star$ or $j \ne j_\star$, $\delta(a_i,b_j) = n+1$. It is easy to verify that $\delta$ is a proper distance function.

Consider now an arbitrary ball in the metric that we just defined. Let $r$ be its radius.
If the ball covers only points in $A$ (or only points in $B$) it is easy to show that it can be covered with at most 3 balls of radius $r/2$. If the ball covers at least one point in both $A$ and $B$, then $r \ge n$, and any ball centered at a point in $A$ of radius $r/2\ge n/2$ covers all points in $A$ and any ball of such radius with center in $B$  covers all points in $B$. Therefore, the doubling dimension of the set equals $\log 3$.

The minimum spanning tree for $A \cup B$ consists of edges $(a_i,a_{i+1})$ and $(b_i,b_{i+1})$ for $i\in[n/2 - 1]$ and of the edge $(a_{i_\star},b_{j_\star})$. Computing the minimum spanning tree requires outputting the edge $(a_{i_\star},b_{j_\star})$. By Lemma~\ref{lem:find_in_array}, the algorithm has to make at least $(n/2)^2/4$ queries to the distance oracle to correctly output the edge with probability greater than $1/2$. This in turn requires $\lceil n^2/(16mT^2) \rceil$ rounds of computation as observed above.
\end{proof}

We assume that $n = \Theta(m\cdot s)$, in which case the number of communication lower bounds has to be $\Omega(n/s)$. In particular, under a common assumption that $s = n^c$, for some constant $c \in (0,1)$, the number of rounds becomes $n^{\Omega(1)}$.

\subsection{Query Lower Bound for Approximate MST for General Metrics}
\label{sec:lb_approx_general}

A similar lower bound holds for general metrics in the black-box distance oracle model, even if we allow for an arbitrarily large constant approximation factor. This explains why allowing just for approximation is not enough and also some additional assumptions are necessary to construct efficient algorithms. In this paper, we assume that the (doubling) dimension of the point set is bounded in order to create efficient algorithms. The result follows directly from the following result of Indyk and the earlier upper bound on the number of queries that can be performed in each round of computation.

\begin{lemma}[\cite{Indyk99}, Section 9]
Finding a $B=O(1)$ approximation to the minimum spanning tree requires $\Omega(n^2/B)$ queries to the distance oracle.
\end{lemma}

\begin{corollary}
Any parallel algorithm computing a $B$-approximation to the minimum spanning tree, where $B=O(1)$, requires $\Omega(n^2/(Bms^2))$ rounds of computation.
\end{corollary}

\section{Acknowledgments}

We thank Mihai Budiu, Frank McSherry, and Omer Reingold for early
discussions about parallel models, including systems such as Dryad and
MapReduce. In particular, we thank Frank McSherry for pointing out the
connection to the streaming with sorting model. We also thank Piotr
Indyk for useful comments on an early draft of the paper.

\bibliographystyle{alpha}
\bibliography{bibliography}

\end{document}